\documentclass[%
 reprint,
superscriptaddress,
 amsmath,amssymb,
 aps,
]{revtex4-2}

\usepackage{dcolumn}% Align table columns on decimal point
\usepackage{bm}% bold math
\usepackage{amsfonts}
\usepackage{color}
\usepackage{amsmath}
\usepackage{mathtools}
\usepackage[colorlinks=true,linkcolor=blue, citecolor=red
]{hyperref}%

\usepackage{tikz}
\usepackage[font=small, figurename=Figure]{caption} 
\usepackage{ragged2e}

\usepackage{physics}

\usepackage{amssymb}
\usepackage{graphicx, caption}
\usepackage{subcaption}
\usepackage{float}
\usepackage{newunicodechar}
\newtheorem{theorem}{Theorem}

\newtheorem{lemma}[theorem]{Lemma}

\newtheorem{proposition}[theorem]{Proposition}

\newenvironment{proof}[1][Proof]{\noindent\textbf{#1}\ \ }{\hfill\rule{0.5em}{0.5em}}
\def\Tr{\operatorname{Tr}}

\def\1{\openone}

\def\swap{\operatorname{SWAP}}

\allowdisplaybreaks

\begin{document}
\title{Optimal input states for quantifying the performance of continuous-variable unidirectional and
bidirectional teleportation}%

\author{Hemant K. Mishra}
\affiliation{Hearne Institute for Theoretical Physics, Department of Physics and Astronomy,
 and Center for Computation and Technology, Louisiana State University, Baton Rouge, Louisiana 70803, USA}
 \affiliation{School of Electrical and Computer Engineering, Cornell University, Ithaca, New York 14850, USA}
\author{Samad Khabbazi Oskouei}
\affiliation{Department of Mathematics, Varamin-Pishva Branch, Islamic Azad University, Varamin, 33817-7489, Iran}
\author{Mark M. Wilde}
\affiliation{Hearne Institute for Theoretical Physics, Department of Physics and Astronomy,
 and Center for Computation and Technology, Louisiana State University, Baton Rouge, Louisiana 70803, USA}
 \affiliation{School of Electrical and Computer Engineering, Cornell University, Ithaca, New York 14850, USA}

%\preprint{APS/123-QED}


\date{\today}% It is always \today, today,

\begin{abstract}
Continuous-variable (CV) teleportation  is a fundamental protocol in quantum information science. A number of experiments have been designed to simulate  ideal teleportation under realistic conditions. In this paper, we detail an analytical approach for determining  optimal input states for quantifying the performance of CV unidirectional and bidirectional teleportation. The metric that we consider for quantifying performance is the energy-constrained channel fidelity between ideal teleportation and its experimental implementation, and along with this, our focus is on determining optimal input states for distinguishing the ideal process from the experimental one.
We prove that, under certain energy constraints, the optimal input state in unidirectional, as well as bidirectional, teleportation is a finite entangled superposition of twin-Fock states saturating the energy constraint. Moreover, we also prove that, under the same constraints, the optimal states are unique; that is, there is no other optimal finite entangled superposition of twin-Fock states.
\end{abstract}

\maketitle

\section{Introduction}

Quantum teleportation  is a foundational protocol in quantum information science that has no classical analogue \cite{bennett1993} (see also \cite{furusawa2007quantum}).
It consists of transmitting an unknown quantum state from one place to another by using shared entanglement and local operations and classical communication (LOCC). 
Quantum teleportation plays an important role in quantum technologies such as quantum information processing protocols \cite{nielsen2001quantum}, quantum computing \cite{gottesman1999demonstrating, knill2001scheme}, and quantum networks~\cite{hermans2022qubit}.
Since the invention of this protocol, various modifications have been proposed, such as probabilistic teleportation \cite{feng2006probabilistic, pati2007probabilistic, yan2010probabilistic}, controlled teleportation \cite{deng2005symmetric, xi2007controlled, zhou2007multiparty, man2007genuine}, and bidirectional teleportation \cite{vaidman1994teleportation, huelga2001quantum, zha2013bidirectional, li2013bidirectional, yan2013bidirectional, fu2014general, hassanpour2016bidirectional, yang2017bidirectional, chen2020bidirectional}.
There has also been significant progress in implementing quantum teleportation in laboratories around the world in the last three decades \cite{liu2020applications}.
Several experiments have implemented the teleportation protocol for simple quantum systems \cite{bouwmeester1997experimental, riebe2004deterministic, ursin2004quantum, wang2015quantum, ren2017, feng2020}, and attempts are being made to extend them to more complex quantum systems \cite{hu2020experimental, luo2019quantum, ma2012experimental, jin2010experimental}.
%Bidirectional quantum teleportation enables transmission of quantum information between Alice
%and Bob in both the directions simultaneously which has attracted much attention in the
%recent years.
%Huelga {\it et al.} discussed teleportation of unitary operations and resources required in the implementation of
%such protocols by using bidirectional teleportation %\cite{huelga2001quantum}.

The first theoretical proposal for quantum teleportation was for two-level quantum systems, also commonly called qubits \cite{bennett1993}. Later, continuous-variable (CV) teleportation was devised as an extension of the original protocol to quantum systems described by infinite-dimensional Hilbert spaces  \cite{vaidman1994teleportation, braunstein1998teleportation}.
This was followed by many experimental implementations of CV teleportation, which include teleportation of collective spins of atomic ensembles \cite{krauter2013deterministic, sherson2006quantum}, polarisation states of photon beams \cite{jin2010experimental}, coherent states \cite{yukawa2008high},  etc.
In standard CV teleportation, the  entangled resource state shared between the sender and receiver, respectively Alice and Bob, is a two-mode squeezed vacuum (TMSV) state. The protocol begins with Alice mixing an unknown input state with her share of the entanglement on a balanced beamsplitter and then performing homodyne detection of complementary quadratures. Based on the classical measurement outcomes received by Alice and subsequently transmitted to Bob, he then performs displacement operations on his share of the TMSV state and recovers an approximation of the original state \cite{braunstein1998teleportation}.

An ideal implementation of CV teleportation in principle allows for perfect transmission of quantum states and hence simulates an ideal quantum channel. However, an ideal implementation also demands the unphysical conditions of noiseless homodyne detection and infinite squeezing in the TMSV state, which is not possible in practice because both noiseless homodyne detection and infinite squeezing require infinite energy. 
Any experimental implementation of CV teleportation accounts for an unideal detection and finite squeezing, which results in an imperfect transmission of quantum states, and hence simulates a noisy quantum channel \cite{braunstein1998teleportation}.
It is therefore important for experimentalists to employ performance metrics, as well as quantify the performance, for any experimental simulation of ideal teleportation.
%It is natural, as well as important, to determine how close is an experimental implementation to the ideal teleportation.

Several works on characterising the performance of experimental implementations of the teleportation protocol have been conducted for finite-dimensional quantum systems in the past few years \cite{christandl2021asymptotic, bang2018fidelity, roy2020rating, studzinski2022efficient, wagner2009performance, wagner2009performance2,HSW22}, including a more recent work on bidirectional teleportation, which benchmarks the performance in terms of normalised diamond distance and channel infidelity for transmission of arbitrary quantum states \cite{siddiqui2020quantifying}.
There have also been many theoretical and experimental works on quantifying the performance of experimental implementations of the CV teleportation protocol. However, most of them study the performance by evaluating specific classes of quantum states, such as coherent states  \cite{hammerer2005quantum, johnson2002continuous, grosshans2001quantum, braunstein2000criteria, johnson2002continuous}, pure single-mode Gaussian states \cite{chiribella2014quantum,kogias2014continuous}, squeezed states \cite{adesso2008quantum}, cat states \cite{seshadreesan2015non}, etc. All such evaluations  are incomplete, in the sense that they test the performance by transmitting specific states rather than arbitrary unknown states.
A true quantifier for CV unidirectional teleportation was given in \cite{SSW20}, which benchmarks the performance of an experimental implementation in terms of the energy-constrained channel fidelity between ideal teleportation and its experimental implementation. We also note here that \cite{SSW20} is foundational for the present paper.
%We extend the work to CV bidirectional teleportation with the help of \cite{SSW20} in laying the groundwork. 

In this paper, we quantify the performance of any experimental implementation of CV unidirectional, as well as bidirectional, teleportation, under certain energy constraints. The performance metric that we consider is the energy-constrained channel fidelity between an ideal teleportation and its experimental implementation. We explicitly find {\it optimal} input states, i.e., quantum states whose output fidelity corresponding to the ideal channel and its experimental approximation is the same as the energy-constrained channel fidelity between the two channels. Our method is purely analytical, employing optimization techniques from multivariable calculus. The optimal states for unidirectional,  as well as bidirectional, teleportation are finite entangled superpositions of twin-Fock states saturating the energy constraint. Furthermore, we prove that the optimal input states are unique; i.e., there is no other optimal finite entangled superposition of twin-Fock states. 

Our results on bidirectional teleportation are also related to one of the most interesting mathematical problems in quantum information theory: the study of additive and multiplicative properties of measures associated with quantum channels \cite{amosov2000additivity, holevo2006additivity, holevo2006multiplicativity, holevo2015gaussian}. Much progress has been made in addressing these additivity issues \cite{cubitt2008counterexamples, hastings2009superadditivity, aubrun2011hastings, fukuda2014revisiting, wilde2011classical}, settling some of the questions posed in \cite{krueger2005some, ruskai2007some}.
The fidelity of quantum states is well known to be multiplicative for tensor-product quantum states \cite{wilde2011classical}.
This induces an inequality for energy-constrained channel fidelity between two tensor-product channels.
As a consequence of our work, we give examples where the induced inequality is strict; that is, our results also imply that the energy-constrained fidelity between the identity channel and an additive-noise channel is strictly sub-multiplicative. 

The rest of our paper is organized as follows. In Section~\ref{prel}, we  review some definitions. We present a derivation of the optimal input state for CV unidirectional teleportation in Section~\ref{unidirectional}, and for CV bidirectional teleportation in Section~\ref{bidirectional}. We show in Section~\ref{mult} that the energy-constrained fidelity between the ideal swap channel and the tensor product of two additive-noise channels  is strictly sub-multiplicative. We then discuss possible extensions and generalizations of the present work in Section~\ref{discussion}. Finally, in Section~\ref{con} we summarize our results and outline questions for future work.

The appendices contain necessary calculations for deriving the results. In Appendix~\ref{appendix:onedirectional}, we provide proofs of some preliminary results required to derive the optimal input state for CV unidirectional teleportation. Similarly, we prove some preliminary results in Appendix~\ref{app:bidirectional} that are used to derive the optimal input state for CV bidirectional teleportation.

\section{Preliminaries}\label{prel}

Let $\mathcal{H}$ be a separable Hilbert space, and let $T$ be an operator acting on $\mathcal{H}$.
The adjoint of $T$ is the unique operator $T^{\dagger}$ acting on $\mathcal{H}$ defined by
$\langle \phi \vert T \vert \psi \rangle = \overline{\langle  \psi \vert T^{\dagger} \vert \phi \rangle}$ for all $\vert \phi \rangle,\vert \psi \rangle \in \mathcal{H};$
 $T$ is said to be self-adjoint if $T=T^{\dagger}$.
 If $\Tr (\sqrt{T^{\dagger}T}) < \infty$ then $T$ is said to be a trace-class operator, and its trace norm is defined as $\|T\|_1 \coloneqq \Tr (\sqrt{T^{\dagger}T})$.
A quantum state is a positive semi-definite, trace-class operator with trace norm equal to one.
We denote by $\mathcal{D}(\mathcal{H})$ the set of all quantum states or density operators acting on $\mathcal{H}$. 
Let $\rho, \sigma \in \mathcal{D}(\mathcal{H})$. The fidelity between $\rho$ and $\sigma$ is defined by \cite{uhlmann1976transition}
\begin{equation}\label{new8_eqn1}
    F(\rho, \sigma) \coloneqq \left\|\sqrt{\rho} \sqrt{\sigma}\right\|_1^2.
\end{equation}
If one of the quantum states is pure, i.e., say $\rho=\vert \psi \rangle\!\langle \psi \vert,$ then $F(\rho, \sigma)=\Tr(\rho \sigma)$. The
sine distance between $\rho$ and $\sigma$ is given by \cite{R02,R03,GLN04,rastegin2006sine}
\begin{equation}
    C(\rho, \sigma) \coloneqq \sqrt{1-F(\rho, \sigma)}.
\end{equation}
The following inequalities relate the fidelity, sine distance, and trace distance \cite[Theorem~1]{fuchs1999cryptographic} 
\begin{equation}\label{new_eqn50}
    1-\sqrt{F(\rho, \sigma)} \leq \frac{1}{2} \left\|\rho - \sigma \right\|_1 \leq C(\rho, \sigma).
\end{equation}

The set of bounded operators on $\mathcal{H}$ forms a $C^*$-algebra under the operator norm, and we denote it by $\mathcal{L}(\mathcal{H})$.
Let $\mathcal{H}_A$ denote the Hilbert space corresponding to a quantum system $A$.
A quantum channel from a quantum system $A$ to a quantum system $B$ is a completely positive, trace preserving linear map from $\mathcal{L}(\mathcal{H}_A)$ to $\mathcal{L}(\mathcal{H}_B)$.
Let $\mathcal{M}_{A \to B}$ and $\mathcal{N}_{A \to B}$ be quantum channels. 
Let $H_A$ be a Hamiltonian corresponding to the quantum system $A,$
and let $R$ denote a reference system.
The energy-constrained channel fidelity between $\mathcal{M}_{A \to B}$ and $\mathcal{N}_{A \to B}$ for $E \in [0, \infty)$ is defined by \cite{Sh19, SWAT18}
\begin{align}\label{new6_eqn1}
    &F_E(\mathcal{M}_{A \to B}, \mathcal{N}_{A \to B}) \coloneqq \nonumber\\
    &\inf_{\substack{\rho_{RA} : \Tr(H_A \rho_A) \leq E}} F(\mathcal{M}_{A \to B}(\rho_{RA}), \mathcal{N}_{A \to B}(\rho_{RA})),
\end{align}
where $\rho_{RA} \in \mathcal{D}(\mathcal{H}_R \otimes \mathcal{H}_A),$ $\rho_A=\Tr_R(\rho_{RA}),$ and it is  implicit that the identity channel $\mathcal{I}_R$ acts on the reference system $R$. Furthermore, the optimization in \eqref{new6_eqn1} is taken over every possible reference system $R$.
Similarly, the energy-constrained sine distance between $\mathcal{M}_{A \to B}$ and $\mathcal{N}_{A \to B}$ for $E \in [0, \infty)$ is defined by \cite{Sh19, SWAT18}
\begin{align}\label{new6_eqn2}
    &C_E(\mathcal{M}_{A \to B}, \mathcal{N}_{A \to B}) \coloneqq \nonumber\\
    &\sup_{\substack{\rho_{RA} : \Tr(H_A \rho_A) \leq E}} C(\mathcal{M}_{A \to B}(\rho_{RA}), \mathcal{N}_{A \to B}(\rho_{RA})).
\end{align}
Although the optimizations in \eqref{new6_eqn1} and \eqref{new6_eqn2} are over arbitrary mixed states and arbitrary reference systems, it suffices to restrict the optimization over pure states such that the reference system $R$ is isomorphic to the channel input system $A$. This is a consequence of purification, the Schmidt decomposition, and data processing \cite[Section 3.5.4]{khatri2020principles}. We thus have
\begin{align}\label{new6_eqn3}
    &F_E(\mathcal{M}_{A \to B}, \mathcal{N}_{A \to B}) = \nonumber\\
    &\inf_{\substack{\phi_{RA} : \Tr(H_A \phi_A) \leq E}} F(\mathcal{M}_{A \to B}(\phi_{RA}), \mathcal{N}_{A \to B}(\phi_{RA})),
\end{align}
\begin{align}\label{new6_eqn4}
    &C_E(\mathcal{M}_{A \to B}, \mathcal{N}_{A \to B}) = \nonumber\\
    &\sup_{\substack{\phi_{RA} : \Tr(H_A \phi_A) \leq E}} C(\mathcal{M}_{A \to B}(\phi_{RA}), \mathcal{N}_{A \to B}(\phi_{RA})),
\end{align}
where the optimizations \eqref{new6_eqn3} and \eqref{new6_eqn4} are taken over pure states $\phi_{RA}$ with reference system $R$ isomorphic to system~$A$.
%One can easily check that
%\begin{equation}
%    C_E(\mathcal{M}_{A \to B}, \mathcal{N}_{A \to B}) = %\sqrt{1-F_E(\mathcal{M}_{A \to B}, \mathcal{N}_{A \to B})}.
%\end{equation}

\section{Optimal input state for CV unidirectional teleportation}\label{unidirectional}

The CV quantum teleportation protocol describes how to transmit an unknown quantum state from Alice to Bob when their systems are in CV modes and they share a prior entangled state known as a {\it resource state} \cite{braunstein1998teleportation}.
In this protocol, Alice mixes the unknown quantum state with her share of
the resource state (TMSV state) and performs homodyne detection.
The homodyne detection destroys the input state on Alice's end.
Alice then communicates the classical outcomes of the detection to Bob, based on which
he performs unitary operations on his share of the resource state to generate an approximation of
the input state.
Let $A$ denote the input mode, and let $B$ denote the output mode.
An ideal teleportation protocol requires noiseless homodyne detection and infinite squeezing in the TMSV state, and it induces the identity channel $\mathcal{I}_{A \to B}$ on the input states \cite{bennett1993,braunstein1998teleportation} (see \cite{PhysRevA.97.062305} for further clarification of the convergence of the protocol to the identity channel).
However, an experimental implementation of CV teleportation has a noisy detector and finite squeezing in the resource state which makes the experimental implementations of teleportation perform less than ideal. It realizes an additive-noise channel $\mathcal{T}_{A\to B}^\xi,$ where the noise parameter $\xi > 0$ encodes unideal detection
and finite squeezing  \cite{braunstein1998teleportation, braunstein1998error}.
The additive-noise channel $\mathcal{T}^{\xi}$
is a composition of the quantum-limited amplifier $\mathcal{A}^{1/\eta}$ with gain parameter $1/\eta$ and the pure-loss channel $ \mathcal{L}^{\eta}$ with transmissivity $\eta,$ where $\eta = 1/(1+\xi)$ \cite{LSHC06, CGH06}.    See \cite[Section~II.B]{sharma2020characterizing} for more details.

%It is natural as well as important to characterise performance of such an experimental setup. 
By taking the performance metric to be the energy-constrained channel fidelity between ideal teleportation and the additive-noise channel, 
the performance of experimental implementations has been studied in \cite{SSW20}.
%Let $E \in [0, \infty)$. 
By choosing the Hamiltonian $H_A$ to be the  photon number operator $\hat{n}_A= \sum_{n=0}^{\infty} n \vert n\rangle\!\langle n\vert_A$, the energy-constrained channel fidelity in \eqref{new6_eqn3} for the identity channel $\mathcal{I}_{A \to B}$ and the additive-noise channel $\mathcal{T}_{A\to B}^\xi$ can be further simplified, as a consequence of phase averaging and joint phase covariance of these channels \cite{SWAT18, SSW20}, as 
\begin{multline}\label{neweqn2}
  F_E(\mathcal{I}_{A\to B}, \mathcal{T}_{A\to B}^\xi)=\\
  \inf_{\psi_{RA}} F(\mathcal{I}_{A\to B}(\psi_{RA}), \mathcal{T}_{A\to B}^\xi(\psi_{RA})),
\end{multline}
where the infimum is taken over pure and entangled superpositions of twin-Fock states $\psi_{RA}=\vert \psi\rangle\!\langle \psi\vert_{RA}$ such that
\begin{equation}
  \vert \psi\rangle_{RA}=\sum_{n=0}^{\infty} \lambda_n \vert n\rangle_R \vert n\rangle_A,
\end{equation}
 $\lambda_n\in \mathbb{R}^+$ for all $n,$ $\sum_{n=0}^{\infty} \lambda_n^2=1,$ and $\sum_{n=0}^{\infty} n\lambda_n^2 \leq E$.
An analytical solution to the energy-constrained channel fidelity in \eqref{neweqn2}, using Karush–Kuhn–Tucker conditions, was given in \cite{SSW20} for small values of $\xi$ and arbitrary values of $E$.
The optimal input state so obtained was
\begin{align}\label{eq1}
   \vert \psi\rangle_{RA}= \sqrt{1-\{E\}} \vert \lfloor E \rfloor\rangle_R \vert \lfloor E \rfloor\rangle_A+\sqrt{\{E\}} \vert \lceil E \rceil\rangle_R \lceil E \rceil\rangle_A,
\end{align}
where $\{E\}\coloneqq E-\lfloor E \rfloor$. Another contribution of \cite{SSW20} was to provide a method, using a combination of numerical and analytical techniques, for finding optimal input states to test the performance of unidirectional CV teleportation under the energy-constrained channel fidelity measure.

In this section, we show that an optimal input state for the energy-constrained channel fidelity \eqref{neweqn2} is a finite entangled superposition of twin-Fock states saturating the energy constraint for arbitrary values of $\xi$ and $E$ satisfying $E \leq (1+\xi)/(1+3\xi),$  and it is given by
\begin{equation}\label{new7_eqn1}
   \vert \psi\rangle_{RA}= \sqrt{1-E} \vert 0\rangle_R \vert 0\rangle_A+\sqrt{E} \vert 1\rangle_R \vert 1\rangle_A.
\end{equation}
Observe that the optimal state in \eqref{new7_eqn1} is the same as that in \eqref{eq1} under the common conditions of $E \leq (1+\xi)/(1+3\xi)$ and small $\xi$. 
Our method also shows that the optimal state in \eqref{new7_eqn1} is unique; i.e., there is no other optimal finite entangled superposition of twin-Fock states for \eqref{neweqn2}.
We emphasize that our method is purely analytical. For, we use optimization techniques from multivariable calculus, and the constraint $E \leq (1+\xi)/(1+3\xi)$ is needed in our analysis in the proof of Proposition~\ref{new3_prop1} that plays a major role in establishing the result. We also note that it is still an open problem to find the optimal state for larger values of $E$ analytically.

In order to compute the energy-constrained channel fidelity between the ideal channel $\mathcal{I}_{A \to B}$ and its experimental implementation $\mathcal{T}_{A \to B}^{\xi},$ we define the {\it $M$-truncated} energy-constrained channel fidelity between $\mathcal{I}_{A \to B}$ and $\mathcal{T}^{\xi}_{A \to B}$ as
\begin{align}\label{new_eqn39}
    F_{E,M}(\mathcal{I}_{A\to B}, \mathcal{T}^{\xi}_{A \to B}) &\coloneqq \inf_{\psi_{RA}} F(\mathcal{I}_{A\to B}(\psi_{RA}), \mathcal{T}_{A\to B}^\xi(\psi_{RA})),
\end{align}
where the infimum is taken over pure states $\psi_{RA}=\vert \psi \rangle\!\langle\psi \vert_{RA}$ of the form
\begin{align}\label{new9_eqn9}
    \vert \psi\rangle_{RA}= \sum_{n=0}^{M} \sqrt{p_n} \vert n\rangle_R \vert n\rangle_A,
\end{align}
such that $p_n \geq 0$ for all $n,$     $\sum_{n=0}^{M} p_n=1,$ and $\sum_{n=0}^{M} n p_n \leq E$.
In Proposition~\ref{new_appendix_B_prop1} in Appendix~\ref{appendix:onedirectional}, we show that 
\begin{multline}\label{new9_eqn10}
    F(\mathcal{I}_{A\to B}(\psi_{RA}), \mathcal{T}^{\xi}_{A \to B}(\psi_{RA})) \\
    \geq \dfrac{1}{(1+\xi)} \left[ \left(\sum_{n=0}^{M}\dfrac{p_n}{(1+\xi)^n}\right)^2+ \left( \sum_{n=1}^{M} \dfrac{p_n\xi}{(1+\xi)^n} \right)^2 \right],
\end{multline}
for every state of the form in \eqref{new9_eqn9}.
Define the real-valued function
\begin{multline}\label{new4_eqn1}
    f_{M,\xi}(p) \coloneqq \\
    \dfrac{1}{(1+\xi)} \left[ \left(\sum_{n=0}^{M}\dfrac{p_n}{(1+\xi)^n}\right)^2+ \left( \sum_{n=1}^{M} \dfrac{p_n\xi}{(1+\xi)^n} \right)^2 \right],
\end{multline}
for all $p \in \mathbb{R}^{M+1}$. By  \eqref{new9_eqn10} and \eqref{new4_eqn1}, we thus have
\begin{align}\label{new9_eqn12}
    F(\psi_{RA}, \mathcal{T}^{\xi}_{A \to B}(\psi_{RA})) \geq  f_{M,\xi}(p).
\end{align}
The minimizer of the function $f_{M,\xi}$ subject to  $p_{n} \geq 0$ for all $n \in\{0,\ldots,M\},$ 
\begin{equation}\label{new9_eqn13}
    \sum_{n=0}^{M} p_{n}=1, \quad  \sum_{n=0}^{M} n p_{n} \leq E, 
\end{equation}
is the unique point given by $p_0=1-E, p_1=E,$ and $p_n=0$ for all $n \geq 2,$ whenever $E \leq (1+\xi)/(1+3\xi)$. See Proposition~\ref{new3_prop1} in Appendix~\ref{appendix:onedirectional}.
It thus follows from \eqref{new9_eqn10} that the optimal input state to the $M$-truncated energy-constrained channel fidelity is unique, and it is given by \eqref{new7_eqn1}, whenever $E \leq (1+\xi)/(1+3\xi)$.
See Lemma~\ref{appendix_B_lemma3} in Appendix~\ref{appendix:onedirectional}.
From the solution to the $M$-truncated energy-constrained channel fidelity, and  the inequality \eqref{new9_eqn11} in Appendix~\ref{appendix:onedirectional}, it follows that the optimal input state in \eqref{neweqn2} is given by \eqref{new7_eqn1},  whenever $E \leq (1+\xi)/(1+3\xi)$. The uniqueness follows from the uniqueness of the optimal state for the $M$-truncated energy-constrained channel fidelity. See Theorem~\ref{main_thm1} in Appendix~\ref{appendix:onedirectional}.

We compare two classes of experimentally relevant quantum states, namely coherent states and TMSV states, with the optimal state under the given energy constraint. See \cite{S17} for further background on CV quantum information. Let $\vert \alpha \rangle$ denote a coherent state, which is given by 
\begin{equation}\label{new6eqn2}
    \vert \alpha \rangle \coloneqq e^{-\frac{|\alpha|^2}{2}} \sum_{n=0}^{\infty} \dfrac{\alpha^{n}}{\sqrt{n!}} \vert n \rangle.
\end{equation}
The energy of the coherent state $|\alpha \rangle$ is $E=|\alpha|^2,$ and its covariance matrix is $I_2,$ the $2 \times 2$ identity matrix. The covariance matrix of $ \mathcal{T}^{\xi}(|\alpha \rangle\!\langle \alpha|)$ is $(1+2\xi)I_2$. 
Let $\psi(\overline{n})_{RA}=|\psi(\overline{n})\rangle\!\langle\psi(\overline{n})|_{RA}$ be the TMSV state given by
\begin{equation}\label{new6eqn1}
    |\psi(\overline{n}) \rangle_{RA} \coloneqq \dfrac{1}{\sqrt{\overline{n}+1}} \sum_{n=0}^{\infty} \sqrt{\left(\dfrac{\overline{n}}{\overline{n}+1} \right)^n} |n \rangle_R |n\rangle_A.
\end{equation}
The energy of its reduced state is $E=\overline{n},$ and its covariance matrix is
\begin{align}
    V_{\psi(\overline{n}) _{RA}}=
    \begin{bmatrix}
    (2\overline{n}+1)I_2 & 2\sqrt{\overline{n}(\overline{n}+1)} \sigma_z \\
    2\sqrt{\overline{n}(\overline{n}+1)} \sigma_z & 
    (2\overline{n}+1)I_2 
    \end{bmatrix},
\end{align}
where $\sigma_z$ is the Pauli-$z$ matrix.
The covariance matrix of $\mathcal{T}^{\xi}(\psi(\overline{n}) _{RA})$ is given by
\begin{align}
    V_{\mathcal{T}^{\xi}(\psi(\overline{n}) _{RA})}=
    \begin{bmatrix}
    (2\overline{n}+1)I_2 & 2\sqrt{\overline{n}(\overline{n}+1)} \sigma_z \\
    2\sqrt{\overline{n}(\overline{n}+1)} \sigma_z & 
    (2\overline{n}+1+2\xi)I_2 
    \end{bmatrix}.
\end{align}
We then have
\begin{align}
    F\left(|\alpha \rangle\!\langle \alpha|, \mathcal{T}^{\xi}(|\alpha \rangle\!\langle \alpha|)\right)& = \dfrac{2}{\sqrt{\operatorname{Det}\left(2(1+\xi)I_2 \right) }}\\
    & =\dfrac{1}{1+\xi} \label{new12_eqn1},
\end{align}
and also,
\begin{align}
    & \!\!\!\!\! F\left(\psi(\overline{n})_{RA}, \mathcal{T}^{\xi}(\psi(\overline{n}) _{RA})\right) \notag \\
    & = \dfrac{2^2}{\sqrt{\operatorname{Det}\left( V_{\psi(\overline{n})_{RA}} +V_{\mathcal{T}^{\xi}(\psi(\overline{n})_{RA})}   \right) }} \\
   & =\dfrac{1}{1+(2\overline{n}+1)\xi} \\
   &= \dfrac{1}{1+(2E+1)\xi}. \label{new12_eqn2}
\end{align}
These fidelity expressions are evaluated using Eq.~$(4.51)$ of \cite{S17}.
\begin{figure}
  \includegraphics[width=0.48\textwidth]{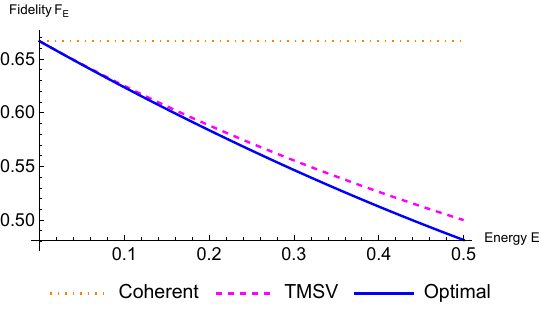}
\caption{\raggedright The graph plots the output fidelity $F_E$ between the ideal channel and an additive-noise channel versus the input energy $E$, corresponding to a coherent state, a TMSV state, and the optimal state (each having input energy $E$). The noise parameter for the additive-noise channel is taken as $\xi=0.5$ and the states have energy $E \in [0,0.5]$. The dotted (orange), dashed (magenta), and solid (blue) lines represent the output fidelity for the coherent state, the TMSV state, and the optimal state, respectively.}
\label{new12fig1}
\end{figure}

In Figure~\ref{new12fig1}, we plot the output fidelity $F_E$ between the ideal channel and an additive-noise channel versus the input energy $E$, corresponding to a coherent state, a TMSV state and the optimal state, with input energy $E \in [0,0.5]$. The noise parameter is taken as $\xi=0.5$.
 In order, the dotted (orange), dashed (magenta), and solid (blue) lines indicate the output fidelity for the coherent state \eqref{new12_eqn1}, the TMSV state \eqref{new12_eqn2}, and the optimal state \eqref{new_eqn41}.
 The graph indicates that coherent and TMSV states are not optimal states in general. Interestingly, however, the TMSV state is very close to being an optimal input state for CV unidirectional teleportation. This was observed in a different regime for the energy constraint, in Figure~2 of \cite{SSW20}.

\section{Optimal input state for CV bidirectional teleportation}

\label{bidirectional}

The CV bidirectional teleportation protocol consists of a two-way transmission of
unknown quantum states between Alice and Bob. One implementation of the protocol, which we consider here, allows for a simulation of two ideal CV unidirectional quantum channels with the help of shared entanglement and LOCC. See \cite{siddiqui2020quantifying} for a discussion of more general implementations. The protocol that we consider here can be thought of as a combination of two CV unidirectional teleportations, one from Alice to Bob and the other from Bob to Alice. 
Let $A$ and $B$ denote the input modes for Alice and Bob, and let $A'$ and $B'$ denote the output modes for Alice and Bob, respectively.
An ideal CV bidirectional teleportation between Alice and Bob is represented by the following unitary swap channel 
\begin{align}\label{swapchannel}
    \mathcal{S}_{AB \to A'B'}(\cdot) \coloneqq \swap (\cdot) \swap^{\dagger},
\end{align}
where the unitary swap operator $\operatorname{SWAP}$ is defined as
\begin{equation}\label{re:swaprelation}
\operatorname{SWAP} \coloneqq \sum_{m, n=0}^{\infty}  \vert m\rangle_{A'}\langle n\vert_A\otimes \vert n\rangle_{B'}\langle m\vert_B.
\end{equation}
Here $\{\vert m\rangle_A\}_{m=0}^\infty $ is the photonic number basis corresponding to the system $A,$ and so on. The swap channel acts on product states by swapping them, i.e., 
\begin{equation}
\mathcal{S}_{AB\to A'B'}(\phi_A \otimes \psi_B)= \psi_{A'} \otimes \phi_{B'}.
\end{equation}
Thus, the swap channel can be thought of as the tensor product of the ideal channels $\mathcal{I}_{A \to B'}$ and  $\mathcal{I}_{B \to A'}$.
An experimental implementation of CV bidirectional teleportation realizes an approximate swap channel given by the tensor product of two additive-noise channels $\mathcal{T}_{A\to B'}^\xi\otimes \mathcal{T}_{B\to A'}^{\xi'}$. 
The Hamiltonian for the composite system $AB$ is the total photon number operator:
\begin{equation}
\hat{n}_{AB} \coloneqq \hat{n}_A \otimes I_B + I_A \otimes \hat{n}_B.
\end{equation}
Given any state $\rho_{AB} \in \mathcal{D}(\mathcal{H}_A \otimes \mathcal{H}_B),$ the inequality
$\Tr (\hat{n}_{AB} \rho_{AB})\leq 2E$
implies that the average photon number in $\rho_{AB}$ over each of the modes $A$ and $B$ is at most $E$.
So, the energy-constrained channel fidelity \eqref{new6_eqn3} between the ideal bidirectional teleportation $\mathcal{S}_{AB \to A'B'}$ and its experimental implementation $\mathcal{T}_{A\to B'}^\xi\otimes \mathcal{T}_{B\to A'}^{\xi'}$ is given by
\begin{multline}\label{new_eqn1}
  F_E(\mathcal{S}_{AB \to A'B'}, \mathcal{T}_{A\to B'}^\xi\otimes \mathcal{T}_{B\to A'}^{\xi'})=  \inf_{\phi_{RAB}:\, \operatorname{Tr}(\hat{n}_{AB}\phi_{AB})\leq 2E}\\
   F\!\left(\mathcal{S}_{AB \to A'B'}(\phi_{RAB}), \mathcal{T}_{A\to B'}^\xi\otimes \mathcal{T}_{B\to A'}^{\xi'}(\phi_{RAB})\right),
\end{multline}
where $\phi_{RAB}$ is a pure state and  $\phi_{AB}=\operatorname{Tr}_R(\phi_{RAB})$. 
As a consequence of the joint phase covariance of $\mathcal{S}_{AB \to A'B'}$ and $\mathcal{T}_{A\to B'}^\xi\otimes \mathcal{T}_{B\to A'}^{\xi'},$ and the arguments given for \eqref{neweqn2}, the infimum in \eqref{new_eqn1} can be recast as
\begin{multline}\label{new_eqn3}
  F_E(\mathcal{S}_{AB \to A'B'}, \mathcal{T}_{A\to B'}^\xi\otimes \mathcal{T}_{B\to A'}^{\xi'})=\\
   \inf_{\psi_{RAB}}
   F\!\left(\mathcal{S}_{AB \to A'B'}(\psi_{RAB}), \mathcal{T}_{A\to B'}^\xi\otimes \mathcal{T}_{B\to A'}^{\xi'}(\psi_{RAB})\right),
\end{multline}
where $\psi_{RAB}=\vert \psi \rangle\!\langle \psi\vert_{RAB}$ is a pure, entangled superposition of twin-Fock states given by 
\begin{align}\label{new7_eqn2}
    \vert \psi \rangle_{RAB}= \sum_{ m,n=0}^{\infty} \lambda_{m,n} \vert m,n\rangle_{R} \vert m,n\rangle_{AB},
\end{align}
such that $\lambda_{m,n} \geq 0$ for all $m$ and $n$, $\sum_{ m,n=0}^{\infty} \lambda_{m,n}^2 = 1$, and 
$\sum_{m, n=0}^{\infty} (m+n) \lambda_{m,n}^2 \leq 2E$.
We further simplify the fidelity expression in \eqref{new_eqn3} as follows. 
Since we are working with CV modes, the systems $A,B,A',B'$ are all isomorphic to each other. So, the energy constrained channel fidelity between $\mathcal{S}_{AB \to A'B'}$ and $\mathcal{T}^{\xi}_{A \to B'} \otimes \mathcal{T}^{\xi'}_{B \to A'}$ must be the same as that of $\mathcal{I}_{A \to A'} \otimes \mathcal{I}_{B \to B'}$ and $\mathcal{T}^{\xi}_{A \to A'} \otimes \mathcal{T}^{\xi'}_{B \to B'}$.
For simplicity of notations, we shall denote any channel $\mathcal{M}_{C \to C'}$ by $\mathcal{M}_{C},$ where $C$ and $C'$ are isomorphic systems.
Recall that an additive-noise channel $\mathcal{T}^{\xi}$
can be written as a composed channel $\mathcal{A}^{1/\eta} \circ \mathcal{L}^{\eta}$ for $\eta=1/(1+\xi)$.
Also, the adjoint of the quantum-limited amplifier is related to the pure-loss channel by $(\mathcal{A}^{1/\eta})^\dagger= \eta \mathcal{L}^{\eta}$ \cite{ISS11}.
For the given pure state in \eqref{new7_eqn2}, we have 
\begin{widetext}
\begin{align}
&F(\mathcal{S}_{AB \to A'B'}(\psi_{RAB}),\mathcal{T}^\xi_{A \to B'}\otimes \mathcal{T}^{\xi'}_{B \to A'}\left(  \psi_{RAB}  \right) ) \nonumber \\
&=  F((\mathcal{I}_A \otimes \mathcal{I}_B)(\psi_{RAB}),(\mathcal{T}_{A}^{\xi}\otimes\mathcal{T}_{B}^{\xi^{\prime}})(\psi_{RAB})) \label{new20eqn5}\\
&  =\operatorname{Tr}[\psi_{RAB}(\mathcal{T}_{A}^{\xi}
\otimes\mathcal{T}_{B}^{\xi^{\prime}})(\psi_{RAB})]\\
&  =\operatorname{Tr}[\psi_{RAB}((\mathcal{A}_{A}^{1/\eta}
\circ\mathcal{L}_{A}^{\eta})\otimes(\mathcal{A}_{B}^{1/\eta^{\prime}%
}\circ\mathcal{L}_{B}^{\eta^{\prime}}))(\psi_{RAB})]\\
&  =\operatorname{Tr}[(\mathcal{A}_{A}^{1/\eta}\otimes\mathcal{A}_{B}^{1/\eta^{\prime}})^{\dag}(\psi_{RAB})(\mathcal{L}_{A}^{\eta
}\otimes\mathcal{L}_{B}^{\eta^{\prime}})(\psi_{RAB})]\\
&  =\eta\eta^{\prime}\operatorname{Tr}[(\mathcal{L}_{A}^{\eta}%
\otimes\mathcal{L}_{B}^{\eta^{\prime}})(\psi_{RAB})(\mathcal{L}%
_{A}^{\eta}\otimes\mathcal{L}_{B}^{\eta^{\prime}})(\psi_{RAB})]\\
&  =\eta\eta^{\prime}\operatorname{Tr}[((\mathcal{L}_{A}^{\eta}%
\otimes\mathcal{L}_{B}^{\eta^{\prime}})(\psi_{RAB}))^{2}]\\
&  =\eta\eta^{\prime}\operatorname{Tr}\left[  \left(  (\mathcal{L}_{A%
}^{\eta}\otimes\mathcal{L}_{B}^{\eta^{\prime}})\left(  \sum_{m,n,m^{\prime
},n^{\prime}=0}^{\infty}\lambda_{m,n}\lambda_{m^{\prime},n^{\prime}%
}|m,n\rangle\!\langle m^{\prime},n^{\prime}|_{R}\otimes|m,n\rangle\!\langle
m^{\prime},n^{\prime}|_{AB}\right)  \right)  ^{2}\right]  \\
&  =\eta\eta^{\prime}\operatorname{Tr}\left[    \left(  \sum
_{m,n,m^{\prime},n^{\prime}=0}^{\infty}\lambda_{m,n}\lambda_{m^{\prime
},n^{\prime}}|m,n\rangle\!\langle m^{\prime},n^{\prime}|_{R}\otimes
\mathcal{L}_{A}^{\eta}(|m\rangle\!\langle m^{\prime}|_{A})\otimes
\mathcal{L}_{B}^{\eta^{\prime}}(|n\rangle\!\langle n^{\prime}|_{B%
})\right)  ^{2}\right]  \\
&  =\eta\eta^{\prime}\operatorname{Tr}\left[
\begin{array}
[c]{c}%
\left(  \sum_{m,n,m^{\prime},n^{\prime}=0}^{\infty}\lambda_{m,n}%
\lambda_{m^{\prime},n^{\prime}}|m,n\rangle\!\langle m^{\prime},n^{\prime}%
|_{R}\otimes\mathcal{L}_{A}^{\eta}(|m\rangle\!\langle m^{\prime}|_{A%
})\otimes\mathcal{L}_{B}^{\eta^{\prime}}(|n\rangle\!\langle n^{\prime
}|_{B})\right)  \times\\
\left(  \sum_{m^{\prime\prime},n^{\prime\prime},m^{\prime\prime\prime
},n^{\prime\prime\prime}=0}^{\infty}\lambda_{m^{\prime\prime},n^{\prime\prime
}}\lambda_{m^{\prime\prime\prime},n^{\prime\prime\prime}}|m^{\prime\prime
},n^{\prime\prime}\rangle\!\langle m^{\prime\prime\prime},n^{\prime\prime\prime
}|_{R}\otimes\mathcal{L}_{A}^{\eta}(|m^{\prime\prime}\rangle\!\langle
m^{\prime\prime\prime}|_{A})\otimes\mathcal{L}_{B}^{\eta^{\prime}%
}(|n^{\prime\prime}\rangle\!\langle n^{\prime\prime\prime}|_{B})\right)
\end{array}
\right]  \\
&  =\eta\eta^{\prime}\sum_{m,n,m^{\prime},n^{\prime}=0}^{\infty}\lambda
_{m,n}^{2}\lambda_{m^{\prime},n^{\prime}}^{2}\operatorname{Tr}\left[
\mathcal{L}_{A}^{\eta}(|m\rangle\!\langle m^{\prime}|_{A})\mathcal{L}%
_{A}^{\eta}(|m^{\prime}\rangle\!\langle m|_{A})\right]
\operatorname{Tr}\left[  \mathcal{L}_{B}^{\eta^{\prime}}(|n\rangle\!\langle
n^{\prime}|_{B})\mathcal{L}_{B}^{\eta^{\prime}}(|n^{\prime}%
\rangle\!\langle n|_{B})\right]  \\
&  =\dfrac{1}{(1+\xi)(1+\xi')}\sum_{m,n,m^{\prime},n^{\prime}=0}^{\infty}\lambda
_{m,n}^{2}\lambda_{m^{\prime},n^{\prime}}^{2}\operatorname{Tr}\left[
\mathcal{L}_{A}^{\frac{1}{1+\xi}}(|m\rangle\!\langle m^{\prime}|_{A})\mathcal{L}%
_{A}^{\frac{1}{1+\xi}}(|m^{\prime}\rangle\!\langle m|_{A})\right]
\operatorname{Tr}\left[  \mathcal{L}_{B}^{\frac{1}{1+\xi'}}(|n\rangle\!\langle
n^{\prime}|_{B})\mathcal{L}_{B}^{\frac{1}{1+\xi'}}(|n^{\prime}%
\rangle\!\langle n|_{B})\right].\label{new20eqn6}
\end{align}

\end{widetext}
Let $p_{m,n}\coloneqq\lambda_{m,n}^2$ for all $m,n \geq 0,$ and
   \begin{align}
      T_{\xi}^{mm'} &\coloneqq     \Tr \left[ \mathcal{L}^{\frac{1}{1+\xi}}_{A} (\vert m \rangle\!\langle m' \vert_{A})  \mathcal{L}^{\frac{1}{1+\xi}}_{A}(\vert m' \rangle\!\langle m \vert_{A}) \right] \label{new_eqn38a},\\
      T_{\xi'}^{nn'} &\coloneqq    \Tr \left[  \mathcal{L}_{B}^{\frac{1}{1+\xi'}}(|n\rangle\!\langle
n^{\prime}|_{B})\mathcal{L}_{B}^{\frac{1}{1+\xi'}}(|n^{\prime}%
\rangle\!\langle n|_{B})\right]\label{new_eqn38b}.
   \end{align}
This gives
   \begin{multline}\label{new_eqn47}
     F\big(\mathcal{S}_{AB \to A'B'}(\psi_{RAB}),\mathcal{T}^\xi_{A \to B'} \otimes \mathcal{T}^{\xi'}_{B \to A'} \left(  \psi_{RAB}  \right) \big)=\\
     \dfrac{1}{(1+\xi)(1+\xi')}\sum_{m, n, m', n'=0}^{\infty}p_{m,n}p_{m' n'}  T_{\xi}^{mm'}  T_{\xi'}^{nn'}.
   \end{multline}
Let $p$ denote the infinite vector $p=(p_{m,n})_{m,n=0}^{\infty}$.  Define
\begin{multline}\label{re:bidf(p)}
 f_{\xi, \xi'}(p)\coloneqq \\ \dfrac{1}{(1+\xi)(1+\xi')} \sum_{m, n, m', n'=0}^{\infty}p_{m,n}p_{m',n'} T_{\xi}^{mm'}  T_{\xi'}^{nn'}.
 \end{multline}
Therefore,  \eqref{new_eqn3} reduces to
the following quadratic optimization problem in terms of an infinite number of variables:
\begin{multline}\label{re:bidirectional}
F_E(\mathcal{S}_{AB \to A'B'}, \mathcal{T}_{A\to B'}^\xi\otimes \mathcal{T}_{B\to A'}^{\xi'})=\\
\begin{cases}
   \begin{aligned}
& \underset{p}{\text{inf}} & & f_{\xi,\xi'}(p) \\
& \text{subject to} & & p_{m,n} \geq 0 \ \ \forall m,n \geq 0, \\
&&& \sum_{ m,n=0}^{\infty} (m+n)p_{m,n} \leq 2E, \\
&&& \sum_{ m,n=0}^{\infty} p_{m,n}=1.
\end{aligned}
\end{cases}
\end{multline}

In this section, we find an optimal input state for the energy-constrained channel fidelity \eqref{re:bidirectional}.  We consider two cases. The first case is when $\xi = \xi'$ and  $2E \leq (1+\xi)/(2+3\xi),$ which corresponds to having identical additive-noise channels in both directions and  states with low energy. The second case is when $\xi'\geq 1$ and $2E \leq \min\{(\xi'^2-1)/(\xi'(3\xi'-1)), (1+\xi)/(2\xi)\},$ which corresponds to experimental implementations with minimum excess noise in one of the quantum channels and low energy states. In both cases, we show that the optimal input state is a finite entangled superposition of twin-Fock states saturating the energy constraint. Furthermore, our method shows that such optimal states  are unique; i.e., there is no other optimal finite entangled superposition of twin-Fock states that achieves this performance. 
Our method is again purely analytical, similar to the unidirectional case, which uses optimization techniques from multivariable calculus. Also, the given constraints on $E,\xi,$ and $\xi'$ are consequences of our analysis in the proofs of Proposition~\ref{new_prop1} and Proposition~\ref{new2_prop1} used to establish the main results.
It still remains to solve \eqref{re:bidirectional} for larger values of $E$; we note here that numerical solutions can be obtained using truncation.

In order to find the solution to the energy-constrained channel fidelity \eqref{re:bidirectional}, we define the $M$-truncated energy-constrained channel fidelity between $\mathcal{S}_{AB\to B'A'}$ and $\mathcal{T}_{A\to B'}^\xi\otimes \mathcal{T}_{B\to A'}^{\xi'}$ as
\begin{multline}\label{new_eqn30}
    F_{E,M}(\mathcal{S}_{AB\to B'A'}, \mathcal{T}_{A\to B'}^\xi\otimes \mathcal{T}_{B\to A'}^{\xi'})\coloneqq \\  \inf_{\psi_{RAB}} 
   F\!\left(\mathcal{S}_{AB\to B'A'}(\psi_{RAB}), \mathcal{T}_{A\to B'}^\xi\otimes \mathcal{T}_{B\to A'}^{\xi'}(\psi_{RAB})\right),
\end{multline}
where the infimum is taken over pure bipartite states $\vert \psi \rangle_{RAB}= \sum_{m,n=0}^{M} \sqrt{p_{m,n}} \vert m,n\rangle_R \vert m,n\rangle_{AB}$ where $p_{m,n} \geq 0$ for all $m,n,$ 
\begin{equation}
    \sum_{m,n=0}^{M} p_{m,n}=1, \hspace{0.5cm}  \sum_{m,n=0}^{M} (m+n) p_{m,n} \leq 2E.
\end{equation}
It turns out that the energy-constrained channel fidelity \eqref{re:bidirectional} can be bounded from above and from below by the $M$-truncated energy-constrained channel fidelity
\eqref{new_eqn30}, which gives a way to approximate the former in terms of the latter. See Proposition~\ref{new_prop2} of Appendix~\ref{app:bidirectional}.

We show in Lemma~\ref{new_lem2} of Appendix~\ref{app:bidir-opt-state-xi'-equal-xi} that the unique optimal input state for the $M$-truncated energy-constrained channel fidelity is given by
\begin{multline}\label{new5_eqn1}
        \vert \psi \rangle_{RAB} = \sqrt{1-2E} \vert 0,0\rangle_{R} \vert 0,0\rangle_{AB}  \\
        +\sqrt{E} \vert 0,1\rangle_{R} \vert 0,1\rangle_{AB}+\sqrt{E}\vert 1,0\rangle_R \vert 1,0\rangle_{AB},
\end{multline}
whenever $\xi=\xi'$ and $2E \leq (1+\xi)/(2+3\xi)$.  We then use the bounds in \eqref{new10_eqn1} to show that \eqref{new5_eqn1} is also an optimal input state for the energy-constrained channel fidelity \eqref{re:bidirectional} under the same conditions $\xi=\xi'$ and $2E \leq (1+\xi)/(2+3\xi),$ and it is unique, as described earlier. See Theorem~\ref{main_thm2} in Appendix~\ref{app:bidir-opt-state-xi'-equal-xi} for a detailed proof.
Using similar methods, we also find an optimal input state for \eqref{re:bidirectional} in the case $\xi' \geq 1$  and $2E \leq \min\{(\xi'^2-1)/(\xi'(3\xi'-1)), (1+\xi)/(2\xi)\}$.
The optimal input state in this case is
\begin{multline}\label{new5_eqn2}
        \vert \psi \rangle_{RAB} = \sqrt{1-2E} \vert 0,0\rangle_{R} \vert 0,0\rangle_{AB}   \\
         +\sqrt{2E-p_{E}} \vert 0,1\rangle_{R} \vert 0,1\rangle_{AB} +\sqrt{p_{E}}\vert 1,0\rangle_R \vert 1,0\rangle_{AB},
\end{multline}
where $p_{E}$ is given by \eqref{new2_eqn10} in Appendix~\ref{app:bidir-opt-state-xi'-xi}.
Again, the state in \eqref{new5_eqn2} is unique. See Theorem~\ref{main_thm3} in Appendix~\ref{app:bidir-opt-state-xi'-xi}.

We compare the optimal state with a tensor product of two coherent states $|\alpha \rangle \otimes |\alpha \rangle$ given by \eqref{new6eqn2} and a tensor product of two TMSV states $\psi(\overline{n})_{RA} \otimes \psi(\overline{n})_{RA}$ given by \eqref{new6eqn1} under the given conditions of Theorem~\ref{main_thm2}.
We find that
\begin{align}
    &F\left(\mathcal{S}\left(|\alpha \rangle\!\langle \alpha| \otimes |\alpha \rangle\!\langle \alpha| \right), \mathcal{T}^{\xi} \otimes \mathcal{T}^{\xi}\left(|\alpha \rangle\!\langle \alpha| \otimes |\alpha \rangle\!\langle \alpha| \right)\right) \nonumber \\ 
    &\hspace{0.5cm} = \left(F\left(|\alpha \rangle\!\langle \alpha|, \mathcal{T}^{\xi}(|\alpha \rangle\!\langle \alpha|)\right) \right)^2 \\
    &\hspace{0.5cm} = \dfrac{1}{(1+\xi)^2}.\label{new20eqn1}
\end{align}
Also, 
\begin{align}
    &F\left(\mathcal{S}\left(\psi(\overline{n})_{RA}  \otimes \psi(\overline{n})_{RA}  \right), \mathcal{T}^{\xi} \otimes \mathcal{T}^{\xi}\left(\psi(\overline{n})_{RA}  \otimes \psi(\overline{n})_{RA}  \right)\right) \nonumber \\ 
    &\hspace{0.5cm} = F\left(\psi(\overline{n})_{RA} , \mathcal{T}^{\xi}(\psi(\overline{n})_{RA} )\right)^2 \\
    &\hspace{0.5cm} = \dfrac{1}{\left(1+(2E+1)\xi \right)^2}.\label{new20eqn2}
\end{align}

Figure~\ref{new12fig2} plots the output fidelity $F_E$ between the ideal swap channel and the tensor product of two identical additive-noise channels versus the input energy $E$, corresponding to a tensor product of coherent, a tensor product of TMSV states and the optimal state, with input energy $E\in [0,0.5]$. The noise parameter is taken as $\xi=0.5$.
In order, the dotted (orange), dashed (magenta), and solid (blue) lines indicate the output fidelity corresponding to the coherent state \eqref{new20eqn1}, the TMSV state \eqref{new20eqn2}, and the optimal state \eqref{new10_eqn2}.
Similar to unidirectional teleportation, neither coherent nor TMSV states are optimal. However, we have proven here that entanglement between the channel uses and a reference system has the optimal performance in distinguishing the two identity channels from the two additive-noise channels (recall \eqref{new5_eqn1}).

%Another interesting state that we compare with the optimal state is a simultaneous CV counterpart of the $GHZ$ and the $W$ states of three qubits given in \cite{adesso2006multipartite}, called the $GHZ/W$ state. Let $\phi$ be the GHZ/W state with zero mean and energy $E=a-\frac{1}{2}$. Its covariance matrix is given by
%\begin{align}
%    V_{\phi}= 
%    \begin{bmatrix}
%      \alpha & \varepsilon & \varepsilon \\
%      \varepsilon & \alpha & \varepsilon \\
%      \varepsilon & \varepsilon & \alpha
%%\end{align}
%where $\alpha = \operatorname{diag}(a, a),$  $\varepsilon = \operatorname{diag}(\varepsilon^+, \varepsilon^-),$ and
%\begin{equation}
%    \varepsilon^{\pm} = \dfrac{a^1-1 \pm \sqrt{(a^2-1)(9a^2-1)}}{4a}.
%\end{equation}

\begin{figure}
  \includegraphics[width=0.48\textwidth]{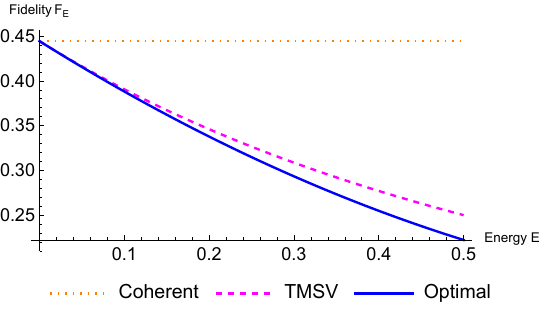}
\caption{\raggedright The graph plots the output fidelity $F_E$ between the ideal swap channel and the tensor product of two identical additive-noise channels versus input energy $E$, corresponding to a tensor product of coherent states, a tensor product of TMSV states, and the optimal state (each having input energy $E$). The noise parameter for the additive-noise channels is taken as $\xi=0.5$ and the states have energy $E \in [0,0.5]$. The dotted (orange), dashed (magenta), and solid (blue) lines represent the output fidelity for the coherent state, the TMSV state, and the optimal state, respectively.}
\label{new12fig2}
\end{figure}

\section{Multiplicativity of energy-constrained channel fidelity}\label{mult}

Let $\mathcal{M}_{i}$ and $\mathcal{N}_{i}$ be quantum channels from quantum subsystem $A_i$ to subsystem $B_i$, and let $\rho_i$ and $\sigma_i$ be quantum states in $A_i$ for $i \in \{1,\ldots, n\}$. The multiplicative property of the fidelity of tensor-product states  follows from the definition in \eqref{new8_eqn1}, i.e.,
\begin{align}
    F\!\left(\bigotimes_{i=1}^n \rho_i,\bigotimes_{i=1}^n \sigma_i \right) 
    = \prod_{i=1}^n F(\rho_i, \sigma_i).
\end{align}
This implies that
\begin{multline}\label{new8_eqn3}
    F\!\left(\left(\bigotimes_{i=1}^n \mathcal{M}_i \right)\left(\bigotimes_{i=1}^n \rho_i\right), \left(\bigotimes_{i=1}^n \mathcal{N}_i \right)\left(\bigotimes_{i=1}^n \sigma_i\right) \right) \\= \prod_{i=1}^n F(\mathcal{M}_i(\rho_i), \mathcal{N}_i(\sigma_i)).
\end{multline}
The energy-constrained channel fidelity between the tensor-product channels $\bigotimes_{i=1}^n \mathcal{M}_i$ and $\bigotimes_{i=1}^n \mathcal{N}_i$ is 
\begin{multline}\label{new8_eqn2}
    F_E\left(\bigotimes_{i=1}^n \mathcal{M}_i, \bigotimes_{i=1}^n \mathcal{N}_i\right) = 
    \inf_{\substack{\psi_{RA_1\cdots A_n}: \Tr(\hat{H}_n \psi_{A_1\cdots A_n}) \leq nE}} \\ F\!\left(\left(\bigotimes_{i=1}^n \mathcal{M}_i\right)(\psi_{RA_1\cdots A_n}), \left(\bigotimes_{i=1}^n \mathcal{N}_i\right)(\psi_{RA_1\cdots A_n})\right).
\end{multline}
In \eqref{new8_eqn2},  $\hat{H}_n$ is the Hamiltonian operator acting on the composite system $A_1\cdots A_n$ given by
\begin{equation}
    \hat{H}_n \coloneqq H_1 \otimes \cdots \otimes I_n+ \cdots + I_1 \otimes \cdots \otimes H_n,
\end{equation}
where $H_1,\ldots,H_n$ are Hamiltonian operators, and $I_1,\ldots, I_n$ are the identity operators for the subsystems $A_1 \cdots A_n$, respectively.
The reference system $R$ can be taken to be $R=R_1\cdots R_n,$ where each reference sub-system $R_i$ is isomorphic to the input subsystem $A_i$.
By taking $\psi_{RA_1\cdots A_n}$ as the $n$-tensor-product state
\begin{equation}
    \psi_{RA_1\cdots A_n} = \psi_{R_1A_1} \otimes \cdots \otimes \psi_{R_nA_n},
\end{equation}
we know from \eqref{new8_eqn3} that
\begin{align}\label{new8_eqn4}
    F\!\left(\left(\bigotimes_{i=1}^n \mathcal{M}_i \right)(\psi_{RA_1\cdots A_n}), \left(\bigotimes_{i=1}^n \mathcal{N}_i \right)(\psi_{RA_1\cdots A_n}) \right)\nonumber \\= \prod_{i=1}^n F(\mathcal{M}_i(\psi_{R_iA_i}), \mathcal{N}_i(\psi_{R_i A_i})).
\end{align}
If each state $\psi_{R_iA_i}$ satisfies the energy constraint $\Tr (H_i \psi_{A_i}) \leq E,$ then $\psi_{RA_1\cdots A_n}$ satisfies the energy constraint in \eqref{new8_eqn2}. Thus, by taking the appropriate infimum on each side of \eqref{new8_eqn4}, we get
\begin{align}\label{new8_eqn5}
    F_E\left(\bigotimes_{i=1}^n \mathcal{M}_i, \bigotimes_{i=1}^n \mathcal{N}_i\right) \leq \prod_{i=1}^n F_E(\mathcal{M}_i, \mathcal{N}_i).
\end{align}
By using the fact that $\mathcal{S}_{AB \to A'B'} \equiv \mathcal{I}_{A \to B'} \otimes \mathcal{I}_{B \to A'},$ from \eqref{new8_eqn5} we get
\begin{multline}\label{new8_eqn6}
    F_E(\mathcal{S}_{AB\to A'B'}, \mathcal{T}_{A \to B'}^{\xi}\otimes \mathcal{T}_{B \to A'}^{\xi'}) \leq \\ 
     F_E(\mathcal{I}_{A\to B'}, \mathcal{T}_{A \to B'}^{\xi})   F_E(\mathcal{I}_{B\to A'},\mathcal{T}_{B \to A'}^{\xi'}).
\end{multline}

The {\it sub-multiplicative} property \eqref{new8_eqn6} is well known. An interesting consequence of Theorems \ref{main_thm1} and \ref{main_thm2} is that the inequality in \eqref{new8_eqn6} can be strict. In particular, we have
\begin{multline}\label{new8_eqn8}
      F_E(\mathcal{S}_{AB \to A'B'}, \mathcal{T}_{A \to B'}^{\xi}\otimes \mathcal{T}_{B \to A'}^{\xi}) < \\ F_E(\mathcal{I}_{A\to B'}, \mathcal{T}_{A \to B'}^{\xi})   F_E(\mathcal{I}_{B\to A'},\mathcal{T}_{B \to A'}^{\xi}),
\end{multline}
whenever $2E \leq (1+\xi)/(2+3\xi)$. This is illustrated in Figure~\ref{new8_plot1}, which plots the difference between the right-hand and the left-hand side of \eqref{new8_eqn8} for $\xi=0.5$. Thus, a consequence of our result on bidirectional teleportation is that the energy-constrained channel fidelity is not multiplicative, i.e., entanglement between a reference and the channel inputs provides a benefit in distinguishing an identity channel from an additive-noise channel, when there is an energy constraint.

\begin{figure}
\centering
  \includegraphics[width=0.48\textwidth]{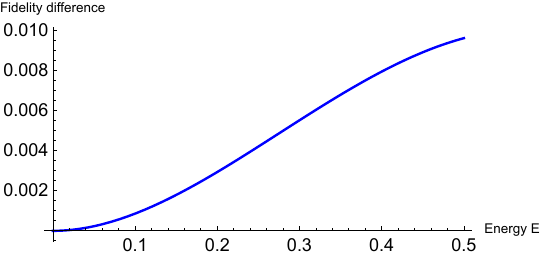}
  \label{new8_plot1b}
\caption{\raggedright The graph plots the difference between the right-hand and the left-hand side of \eqref{new8_eqn8} versus the energy $E\in [0, 0.5]$ for $\xi=0.5$.}\label{new8_plot1}
\end{figure}
%Another justification for the strict inequality in  \eqref{new8_eqn8} is as follows. A necessary condition for the equality in \eqref{new8_eqn8}, as a consequence of the uniqueness of optimal states, is that the optimal state in \eqref{new5_eqn1} is a tensor product of pure states. However, the necessary condition is not satisfied by the state $\psi_{RAB}$ in \eqref{new5_eqn1} because the reduced state $\psi_{AB}$ is not pure:
%\begin{align}
%    \Tr(\psi^2_{AB})=1-4E(1-E)<1.
%\end{align}

\section{Discussion}

\label{discussion}

A characterisation of the performance of any experimental implementation of CV unidirectional teleportation in terms of the energy-constrained channel fidelity is known for arbitrary values of the noise parameter $\xi$ and the energy constraint $E$ \cite{SSW20}.  In our work, we characterise the performance of an experimental implementation of CV bidirectional teleportation under certain conditions on the noise parameters $\xi$ and $\xi'$, as well as the energy constraint $E$. It still remains to solve the problem in the most general case, i.e., for arbitrary values of the noise parameters and every energy constraint. 

It is also an interesting question for future work to consider the following quantity, and which states achieve the optimal value:
\begin{equation}
    F_E(\mathcal{I}^{\otimes n}, (\mathcal{T}^\xi)^{\otimes n}).
    \label{eq:energy-constr-fid-n-parties}
\end{equation}
This quantity applies to a generalization of CV bidirectional teleportation in which there are $n$ parties involved, and the $i$th party is trying to communication quantum information to the $i+1$st party for $i\in\{1,\ldots, n-1\}$ and the $n$th party is trying to do so to the $1$st party. Based on the findings of our paper, in particular the structures of the optimal states \eqref{new7_eqn1} and \eqref{new5_eqn2} for unidirectional and bidirectional teleportation, we suspect that for a sufficiently low energy constraint $E$, the following state optimizes the value in \eqref{eq:energy-constr-fid-n-parties}:
\begin{multline}\label{new20eqn7}
    \sqrt{1-nE}\lvert0 \rangle_R \lvert 00  \cdots 00 \rangle_{A_1A_2 \cdots A_{n-1}A_n}\\ 
    + \sqrt{E} \lvert 1 \rangle_R \lvert 10  \cdots 00 \rangle_{A_1A_2 \cdots A_{n-1}A_n}  \\
    + \sqrt{E}\lvert 2 \rangle_R \lvert 01 \cdots 00 \rangle_{A_1A_2 \cdots A_{n-1}A_n}\\
    \hspace{1cm} \vdots \\
    + \sqrt{E}\lvert n-1 \rangle_R \lvert 00  \cdots 10 \rangle_{A_1A_2 \cdots A_{n-1}A_n} \\
    + \sqrt{E}\lvert n \rangle_R \lvert 00 \cdots 01 \rangle_{A_1A_2 \cdots A_{n-1}A_n}.
\end{multline}
Note that the reduced state on systems $A_1 \cdots A_n$ has total energy $nE$.
To support the conjecture for $n=3$, we compare the output fidelity corresponding to a tensor product of coherent states, a tensor product of TMSV states, and the conjectured optimal state, and show that the conjectured optimal state behaves the same as the optimal states for unidirectional and bidirectional teleportation.

\begin{figure}
  \includegraphics[width=0.48\textwidth]{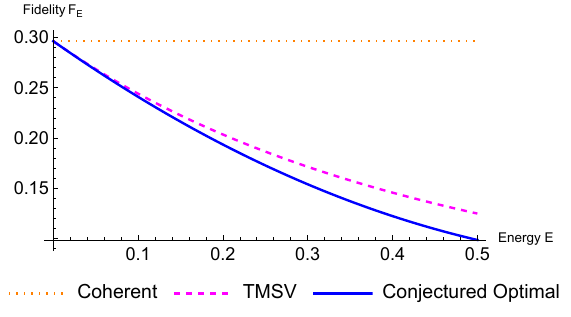}
\caption{\raggedright 
The graph plots the output fidelity $F_E$ between the tensor product of three identity channels and the tensor product of three identical additive-noise channels versus the input energy $E$, corresponding to a tensor product of coherent states, a tensor product of TMSV states, and the conjectured optimal state (each having the input energy $E$). The noise parameter for the additive-noise channels is taken as $\xi=0.5$ and the states have energy $E \in [0,0.5]$. The dotted (orange), dashed (magenta), and solid (blue) lines represent the output fidelity for the coherent state, the TMSV state, and the optimal state, respectively.}
\label{new12fig3}
\end{figure}
For the coherent state $\alpha=|\alpha\rangle\!\langle \alpha|$ given by \eqref{new6eqn2}, we have
\begin{align}
    F\left(\mathcal{I}^{\otimes 3}\left(\alpha^{\otimes 3}  \right), (\mathcal{T}^\xi)^{\otimes 3}\left(\alpha^{\otimes 3} \right)\right) &= \left(F\left(\alpha, \mathcal{T}^\xi(\alpha)\right) \right)^3 \\
    &=\dfrac{1}{(1+\xi)^3}.\label{new20eqn3}
\end{align}
For the TMSV state $\psi(\overline{n})$ given by \eqref{new6eqn1}, we have
\begin{align}
    &F\left(\mathcal{I}^{\otimes 3}\left(\psi(\overline{n})^{\otimes 3} \right), (\mathcal{T}^\xi)^{\otimes 3}\left(\psi(\overline{n})^{\otimes 3} \right)\right) \nonumber \\
    &\hspace{1cm} = \left(F\left(\psi(\overline{n}), \mathcal{T}^{\xi}(\psi(\overline{n}))\right) \right)^3 \\  &\hspace{1cm}=\dfrac{1}{(1+(1+2E)\xi)^3}.\label{new20eqn4}
\end{align}
Let $\phi$ be the conjectured optimal state given by \eqref{new20eqn7} for $n=3$. By following the same arguments given in \eqref{new20eqn5}-\eqref{new20eqn6}, we get
\begin{align}
    &F\left(\mathcal{I}^{\otimes 3}\left(\phi^{\otimes 3} \right), (\mathcal{T}^\xi)^{\otimes 3}\left(\phi^{\otimes 3} \right)\right) \nonumber \\
    &\hspace{1cm} = \dfrac{1}{(1+\xi)^3} \left[1-6\left(\dfrac{\xi }{1+\xi}\right)E+12\left(\dfrac{\xi }{1+\xi}\right)^2E^2\right].\label{new20eqn8}
\end{align}

 In Figure~\ref{new12fig3}, we plot the output fidelity $F_E$ corresponding to  \eqref{new20eqn3}, \eqref{new20eqn4}, and \eqref{new20eqn8} versus the input energy $E \in [0,0.5]$ for $\xi=0.5$. The plot shows similar behavior of the conjectured optimal state as the plots for the optimal states for unidirectional and bidirectional teleportation, which is expected.

% CV bidirectional quantum teleportation can be theoretically generalised to $N$ users (i.e., multidirectional teleportation). In this protocol, there are $N \geq 2$ users sharing an $N$-partite entangled state, and any two parties can perform bidirectional teleportation between them \cite{van2000multipartite}. 
% A natural representation of any experimental implementation of CV $N$-partite teleportation is as a tensor product of $\binom{N}{2}$ pairs of additive-noise channels $\bigotimes_{i,j} (\mathcal{T}^{\xi_i} \otimes \mathcal{T}^{\xi_j}).$ We can thus characterise performance of the experimental implementation in terms of its energy constrained channel fidelity with the tensor product of $N$ identity channels. We suspect that an optimal state for $N$-partite CV teleportation is given by
% \begin{align}
%     \vert \psi & \rangle_{RA_1\cdots A_N}\nonumber\\
%     &=\lvert0 \rangle_R \lvert 00 00 \cdots 00 \rangle_{A_1A_2A_1A_3 \cdots A_{N-1}A_N} 
%     \\
%     &+ \lvert 1 \rangle_R \lvert 01 00 \cdots 00 \rangle_{A_1A_2A_1A_3 \cdots A_{N-1}A_N}  \\
%     &+ \lvert 2 \rangle_R \lvert 10 00 \cdots 00 \rangle_{A_1A_2A_1A_3 \cdots A_{N-1}A_N}\\
%     &\hspace{1cm} \vdots \\
%     &+ \lvert 2\binom{N}{2}-1 \rangle_R \lvert 00 00 \cdots 01 \rangle_{A_1A_2A_1A_3 \cdots A_{N-1}A_N} \\
%     &+ \lvert 2\binom{N}{2} \rangle_R \lvert 00 00 \cdots 10 \rangle_{A_1A_2A_1A_3 \cdots A_{N-1}A_N}
% \end{align}
% under certain non-trivial conditions on the noise parameters of the additive-noise channels and the energy constraint.

\section{Conclusion}

\label{con}

We have characterised the performance of CV unidirectional and bidirectional teleportation in terms of the energy-constrained channel fidelity between ideal CV teleportation and its experimental approximation. Through a purely analytical method, using optimization techniques from multivariable calculus, we explicitly determined the optimal input state for CV unidirectional, as well as bidirectional, teleportation. We showed that, in both the protocols, the optimal input state is a finite entangled superposition of twin-Fock states. Furthermore, the optimal states are unique; i.e., there is no other optimal finite entangled superposition of twin-Fock states.
As an application of our results, we have shown that the energy-constrained channel fidelity of two tensor-product channels is strictly sub-multiplicative.

Another metric to quantify performance of any experimental implementation of CV unidirectional and bidirectional protocol is the Shirokov--Winter energy-constrained diamond distance \cite{S18,W17} between ideal CV teleportation and its experimental approximation. We leave the study of this quantity for future work. Additionally, it remains an intriguing open question to determine optimal input states for multidirectional CV teleportation, which we also leave for future work.

%%%%%%%%%%%%%
%%%%%%%%%%%%%%%
%%%%%%%%%%%%%%%%
%%%%%%%%%%%%

\begin{acknowledgments}

HKM and MMW acknowledge support from the National Science Foundation under Grant No.~2014010. We thank the participants of QuILT Day (May 2022) at Tulane University for feedback on our work, which helped to improve it.
HKM thanks Komal Malik for helping with a technical calculation.
\end{acknowledgments}

\bibliographystyle{unsrt}
\bibliography{RefBT}

\appendix

\onecolumngrid
%%%%%%
%%%%%%
%%%%%%
%\section{Quantum states and channels}\label{app:quantumgaussian}

%WE NEED TO YOUR NOTES FROM \cite[subsection 1 in Appendix A]{SSW20}.

\section{Optimal state for CV unidirectional teleportation}

\label{appendix:onedirectional}

We shall use the following Kraus representation for the pure-loss channel \cite{ISS11}:
\begin{equation}\label{new3_eqn1}
\mathcal{L}^\eta_{A \to B}(\vert m\rangle\!\langle m'\vert_A)=\sum_{k=0}^{\min\{m, m'\}} \sqrt{\binom{m}{k} \binom{m'}{k}}\eta^{\frac{1}{2}(m+m'-2k)} (1-\eta)^k \vert m-k\rangle\!\langle m'-k\vert_B.
\end{equation}

\begin{proposition}\label{new_appendix_B_prop1}
Let $p=(p_0,p_1,p_2,\ldots)$ be an arbitrary infinite-dimensional probability vector, and let $\psi_{RA}$ be a pure state given by
$\vert \psi \rangle_{RA}= \sum_{n=0}^{\infty} \sqrt{p_n} \vert n\rangle_R \vert n\rangle_A$.
Then
\begin{equation}\label{appendix_B_eqn2}
    F(\mathcal{I}_{A\to B}(\psi_{RA}), \mathcal{T}^{\xi}_{A \to B}(\psi_{RA})) \geq \dfrac{1}{(1+\xi)} \left[ \left(\sum_{n=0}^{M}\dfrac{p_n}{(1+\xi)^n}\right)^2+ \left( \sum_{n=1}^{M} \dfrac{p_n\xi}{(1+\xi)^n} \right)^2 \right].
\end{equation}
The inequality in \eqref{appendix_B_eqn2} is saturated if $p_n = 0$ for all $n \geq 2$.
\end{proposition}
\begin{proof}
Let $\psi_A = \Tr_R (\psi_{RA})=\sum_{n=0}^{\infty} p_n \vert n\rangle\!\langle n \vert_A$. From \eqref{new3_eqn1} we get
\begin{align}\label{new9_eqn2}
    \mathcal{L}_{A \to B}^{1-\eta}(\psi_{A})= \sum_{n=0}^{\infty} p_n \sum_{k=0}^{n}  \binom{n}{k}(1-\eta)^k \eta^{n-k} \vert k\rangle\!\langle k\vert_B.
\end{align}
For $\eta=1/(1+\xi),$ we have
\begin{eqnarray}
F(\mathcal{I}_{A\to B}(\psi_{RA}), \mathcal{T}_{A \to B}^{\xi}(\psi_{RA})) &=& \eta \Tr((\mathcal{L}_{A \to B}^{1-\eta}(\psi_{A}))^2)\label{re:onefidelity} \\
    &\geq & \eta \Tr\left[\left(\sum_{n=0}^{M} p_n \sum_{k=0}^{n}  \binom{n}{k}(1-\eta)^k \eta^{n-k} \vert k\rangle\!\langle k\vert_B \right)^2 \right] \label{re:onesigmaprev} \\
   &=& \eta \Tr\left[\left(\sum_{k=0}^{M} \left\{\sum_{n=k}^{M} p_n \binom{n}{k}(1-\eta)^k \eta^{n-k}\right\} \vert k\rangle\!\langle k\vert_B \right)^2   \right]\label{re:onesigma} \\
   &=& \eta \sum_{k=0}^{M} \left\{\sum_{n=k}^{M} p_n \binom{n}{k}(1-\eta)^k \eta^{n-k}\right\}^2 \label{new9_eqn3}  \\
      &\geq& \eta \sum_{k=0}^{1} \left\{\sum_{n=k}^{M} p_n \binom{n}{k}(1-\eta)^k \eta^{n-k}\right\}^2 \label{new5_eqn4}\\
      &=& \eta\left(\sum_{n=0}^{M}p_n \eta^n\right)^2+ \eta\left( \sum_{n=1}^{M} p_n n(1-\eta)\eta^{n-1}\right)^2 \label{new9_eqn4} \\
      &\geq&  \eta\left(\sum_{n=0}^{M}p_n \eta^n\right)^2+ \eta\left( \sum_{n=1}^{M} p_n (1-\eta)\eta^{n-1}\right)^2 \label{new5_eqn3}\\
      &=&  \dfrac{1}{(1+\xi)} \left[ \left(\sum_{n=0}^{M}\dfrac{p_n}{(1+\xi)^n}\right)^2+ \left( \sum_{n=1}^{M} \dfrac{p_n\xi}{(1+\xi)^n} \right)^2 \right].
\end{eqnarray}
The equality \eqref{re:onefidelity} is given in \cite{SSW20}, the inequality \eqref{re:onesigmaprev} follows by truncating the infinite sum in \eqref{new9_eqn2} to $M,$ the equality \eqref{re:onesigma} is obtained by interchanging the indices of the sum in \eqref{re:onesigmaprev}, and \eqref{new9_eqn3} follows by definition of trace. We get \eqref{new5_eqn4} by truncating the outer sum in \eqref{new9_eqn3} to $1,$ and the inequality \eqref{new5_eqn3} is obtained by replacing $n$ with $1$ in the multiples of $p_n$ in the second term of \eqref{new9_eqn4}. We use the relation $\eta = 1/(1+\xi)$ in the last equality. Furthermore, each inequality is saturated if $p_n=0$ for all $n \geq 2$.
\end{proof}

\begin{proposition}\label{new3_prop1}
The minimum value of $f_{M,\xi}$, as defined in \eqref{new4_eqn1}, subject to $p_{n} \geq 0$ for all  $n\in\{0,\ldots,M\},$
\begin{equation}\label{new4_eqn45}
    \sum_{n=0}^{M} p_{n}=1, \quad  \sum_{n=0}^{M} n p_{n} \leq E, 
\end{equation}
is attained at the unique point $p=(p_n)_{n=0}^M$ given by $p_{0}=1-E, \ p_{1}=E,$ and $p_{n}=0$ for all $n\in \{2,\ldots,M\},$ whenever $E \leq (1+\xi)/(1+3\xi)$.
\end{proposition}
\begin{proof}
We divide the proof into two parts. Let $E_0 \in [0,E]$. In the first part, we show that the minimum of $f_{M,\xi},$ subject to $p_{n} \geq 0$ for all  $n \in \{0,\ldots,M\}$ and  the equality constraints
\begin{equation}\label{new9_eqn6}
    \sum_{n=0}^{M} p_{n}=1, \quad  \sum_{n=0}^{M} n p_{n} = E_0,
\end{equation}  
is uniquely attained at $p_0=1-E_0, p_1=E_0,$ and $p_n=0$ for all $n\in \{2,\ldots, M\}$. In the second part, we show that the minimum value obtained in the first part is a strictly decreasing function of $E_0$.

\emph{Part 1.}
Let $p \in \mathbb{R}^{M+1}$ be any vector satisfying the equality constraints \eqref{new9_eqn6},
which gives
\begin{equation}\label{re:7}
    p_0 = 1-E_0+\sum_{n=2}^{M} (n-1) p_n, \quad p_1 = E_0-\sum_{n=2}^{M} n p_n.
\end{equation}
Substitute the values of $p_0$ and $p_1$ in \eqref{new4_eqn1}, and use the relation $\eta= 1/(1+\xi) \in [0,1]$ to get
\begin{equation}\label{new4_eqn2}
    f_{M,\xi}(p)=\eta\left(1-(1-\eta)E_0+\sum_{n=2}^{M} [n(1-\eta)-1+\eta^n]p_n \right)^2+ \eta (1-\eta)^2\left(E_0 - \sum_{n=2}^{M} (n-\eta^{n-1})p_n \right)^2.
\end{equation}
Let $g:\mathbb{R}^{M-1} \to \mathbb{R}$ be the function defined by
\begin{align}
    g(x_2,\ldots,x_M) \coloneqq\eta\left(1-(1-\eta)E_0+\sum_{n=2}^{M} [n(1-\eta)-1+\eta^n]x_n \right)^2+ \eta (1-\eta)^2\left(E_0 - \sum_{n=2}^{M} (n-\eta^{n-1})x_n \right)^2.
\end{align}
From \eqref{new4_eqn2}, we thus have $f_{M,\xi}(p)=g(p_2,\ldots,p_M)$. We show that $g$ is a strictly increasing function in each variable over 
\begin{equation}
\mathbb{R}^{M-1}_+=\{(x_2,\ldots,x_m): \forall k \in \{2,\ldots, M\}, x_k \geq 0 \}.
\end{equation}
This implies that the unique global minimizer of $g$ over $\mathbb{R}^{M-1}_+$ is at the origin.
For all $k \in \{2,\ldots, M\}$, differentiate $g$ partially with respect to $x_k$. We get
\begin{multline}\label{appendix_B_eqn4}
  \frac{\partial g(x_2,\ldots,x_M)}{\partial x_k}=2\eta(\eta^k-k\eta +k-1) \left[1-E_0+\eta E_0+ \sum_{n=2}^{M} \left[n(1-\eta)-1+\eta^n\right] x_n\right]\\+2\eta (1-\eta)^2(\eta^{k-1}-k)\left[E_0 +\sum_{n=2}^{M}  (\eta^{n-1}-n)x_n\right].
\end{multline}
Simplifying \eqref{appendix_B_eqn4} gives
\begin{multline}\label{appendix_B_eqn6}
  \frac{\partial g(x_2,\ldots,x_M)}{\partial x_k}= 2 \eta \bigg[ \big[k(1-\eta)-(1-\eta^{k})\big]\big[1-(1-\eta) E_0\big]- (k-\eta^{k-1})(1-\eta)^2E_0 \bigg] \\
  +2 \eta  \sum_{n=2}^{M} \bigg[\big[k(1-\eta) -(1-\eta^k)\big]\big[n(1-\eta)-(1-\eta^n)\big]+ (1-\eta)^2(k-\eta^{k-1})(n-\eta^{n-1}) \bigg]x_n.
\end{multline}
The coefficients of $x_n$ are positive because for all $k \geq 2,$
\begin{equation}\label{appendix_B_eqn11}
 k(1-\eta)-(1-\eta^k)=(1-\eta) \sum_{i=0}^{k-1} (1-\eta^i) > 0.
\end{equation}
Also, the remaining term in \eqref{appendix_B_eqn6} is positive:
\begin{align}
    &2 \eta \bigg[ \big[k(1-\eta)-(1-\eta^{k})\big]\big[1-(1-\eta) E_0\big]- (k-\eta^{k-1})(1-\eta)^2E_0 \bigg] \nonumber\\
    &\hspace{1cm} \geq 2 \eta \left[ \big[k(1-\eta)-(1-\eta^{k})\big]\left(1- \frac{1-\eta}{3-2\eta} \right)- \frac{(k-\eta^{k-1})(1-\eta)^2}{3-2\eta} \right] \label{appendix_B_eqn12} \\
   &\hspace{1cm}= \frac{2 \eta}{3-2\eta} \bigg[ \big[k(1-\eta)-1+\eta^{k}\big](2-\eta) -(k-\eta^{k-1})(1-\eta)^2 \bigg] \\
   &\hspace{1cm} = \frac{2 \eta}{3-2\eta} \bigg[k(1-\eta)-1+\eta^{k}+\big[k(1-\eta)-1+\eta^{k}\big](1-\eta) -k(1-\eta)^2+\eta^{k-1}(1-\eta)^2 \bigg] \\
   &\hspace{1cm} = \frac{2 \eta}{3-2\eta} \bigg[k(1-\eta)-1+\eta^{k}+k(1-\eta)^2-(1-\eta)+\eta^{k}(1-\eta) -k(1-\eta)^2+\eta^{k-1}(1-\eta)^2 \bigg] \\
   &\hspace{1cm} = \frac{2 \eta}{3-2\eta} \bigg[k(1-\eta)+\eta^{k}+\eta+\eta^{k-1}(1-\eta) \bigg] \\
   &\hspace{1cm}= \frac{2 \eta}{3-2\eta} \bigg[k(1-\eta)+\eta+\eta^{k-1} \bigg] \\
   & \hspace{1cm} > 0 \label{appendix_B_eqn13}.
\end{align}
The inequality \eqref{appendix_B_eqn12} follows from the fact that $E_0 \leq E \leq (1+\xi)/(1+3\xi) = 1/(3-2\eta),$
and \eqref{appendix_B_eqn13} follows from \eqref{appendix_B_eqn11}.
We thus have $\partial g(x_2,\ldots,x_M)/\partial x_k > 0,$ and hence
$g$ is a strictly increasing function of each of its variables in $\mathbb{R}^{M-1}_+$.

\emph{Part 2.} We know from the first part of the proof that the minimum value of $f_{M, \xi}$ subject to the constraints \eqref{re:7} is
\begin{equation}\label{new9_eqn7}
    f_{M,\xi}(1-E_0,E_0,0,\ldots,0)= \frac{1}{(1+\xi)}\left[ 1-2\left(\dfrac{\xi}{1+\xi} \right)E_0+2\left(\dfrac{\xi}{1+\xi} \right)^2 E_0^2  \right].
\end{equation}
The quadratic polynomial $1-2\left(\xi/(1+\xi) \right)E_0+2\left(\xi/(1+\xi) \right)^2 E_0^2$
is a strictly decreasing function in the interval $[0, (1+\xi)/(2\xi)],$ and from the hypothesis we have $E_0 \leq E \in [0, (1+\xi)/(2\xi)]$. 
Therefore, the minimum value of $f_{M,\xi}$ subject to the constraints \eqref{new4_eqn45} is obtained at $p_0=1-E, p_1=E,$ and  $p_k=0$ for all $k \geq 2$.
\end{proof}
\begin{lemma}\label{appendix_B_lemma3}
The $M$-truncated energy-constrained channel fidelity \eqref{new_eqn39} has the unique optimal state given by
\eqref{new7_eqn1},
whenever $E \leq (1+\xi)/(1+3\xi)$.
Furthermore, 
\begin{equation}
    F_{E,M}(\mathcal{I}_{A\to B}, \mathcal{T}^{\xi}_{A \to B}) = \frac{1}{(1+\xi)}\left[ 1-2\left(\dfrac{\xi}{1+\xi} \right)E+2\left(\dfrac{\xi}{1+\xi} \right)^2 E^2  \right].
\end{equation}
 In particular, $F_{E,M}(\mathcal{I}_{A\to B}, \mathcal{T}^{\xi}_{A \to B})$ is independent of $M$.
\end{lemma}
\begin{proof}
Let us consider any pure bipartite state
$\vert \phi\rangle_{RA}= \sum_{n=0}^{M} \sqrt{p_n} \vert n\rangle_R \vert n\rangle_A$ with $p_n \geq 0$ for all $n \in \{0,\ldots,M\},$ such that 
\begin{equation}
    \sum_{n=0}^{M} p_n=1, \hspace{0.5cm}  \sum_{n=0}^{M} n p_n \leq E.
\end{equation}
From Propositions~\ref{new_appendix_B_prop1} and \ref{new3_prop1}, we have
\begin{equation}\label{appendix_B_eqn7}
    F(\phi_{RA}, \mathcal{T}_{A \to B}^{\xi}(\phi_{RA})) \geq f_{M,\xi}(1-E,E,0,\ldots,0)= \frac{1}{(1+\xi)}\left[ 1-2\left(\dfrac{\xi}{1+\xi} \right)E+2\left(\dfrac{\xi}{1+\xi} \right)^2 E^2  \right].
\end{equation}
By taking the infimum in \eqref{appendix_B_eqn7} over $\phi_{RA},$ and by definition \eqref{new_eqn39}, we get
\begin{equation}\label{appendix_B_eqn9}
    F_{E,M}(\mathcal{I}_{A\to B}, \mathcal{T}^{\xi}_{A \to B}) \geq \frac{1}{(1+\xi)}\left[ 1-2\left(\dfrac{\xi}{1+\xi} \right)E+2\left(\dfrac{\xi}{1+\xi} \right)^2 E^2  \right].
\end{equation}
Also, the inequality in \eqref{appendix_B_eqn7} is saturated for the state in \eqref{new7_eqn1}.
This means that \eqref{appendix_B_eqn9} is actually an equality. The uniqueness of the optimal state follows from Proposition~\ref{new3_prop1}.
\end{proof}

\begin{theorem}\label{main_thm1}
The energy-constrained channel fidelity \eqref{neweqn2} has an optimal input state given by \eqref{new7_eqn1},
whenever $E \leq (1+\xi)/(1+3\xi)$.
Moreover, the optimal state is unique in the sense that there is no other optimal finite entangled superposition of twin-Fock states.  
The value of the energy-constrained channel fidelity is
\begin{equation}\label{new_eqn41}
     F_E(\mathcal{I}_{A\to B}, \mathcal{T}_{A\to B}^\xi) = \frac{1}{(1+\xi)}\left[ 1-2\left(\dfrac{\xi}{1+\xi} \right)E+2\left(\dfrac{\xi}{1+\xi} \right)^2 E^2  \right].
\end{equation}
\end{theorem}
\begin{proof}
From the inequalities (B31) of \cite{SSW20}, we have  
\begin{equation}\label{new9_eqn8}
    1-\left[2 \sqrt{\frac{E}{M+1}} + \sqrt{1-F_{E,M}(\mathcal{I}_{A\to B}, \mathcal{T}_{A\to B}^\xi)} \right]^2 \leq F_E(\mathcal{I}_{A\to B}, \mathcal{T}_{A\to B}^\xi) \leq F_{E,M}(\mathcal{I}_{A\to B}, \mathcal{T}_{A\to B}^\xi).
\end{equation}
By Lemma~\ref{appendix_B_lemma3}, we thus get
\begin{multline}\label{new9_eqn11}
    1-\left[2 \sqrt{\frac{E}{M+1}} + \sqrt{1-\frac{1}{(1+\xi)}\left( 1-2\left(\dfrac{\xi}{1+\xi} \right)E +2\left(\dfrac{\xi}{1+\xi} \right)^2 E^2  \right)} \right]^2 \\ \leq F_E(\mathcal{I}_{A\to B}, \mathcal{T}_{A\to B}^\xi)  \leq \frac{1}{(1+\xi)}\left[ 1-2\left(\dfrac{\xi}{1+\xi} \right)E+2\left(\dfrac{\xi}{1+\xi} \right)^2 E^2  \right].
\end{multline}
\eqref{new_eqn41} is obtained by taking the limit $M \to \infty,$ and the optimal state is given by \eqref{new7_eqn1}, which follows from the proof of Lemma~\ref{appendix_B_lemma3}.

Any finite entangled superposition of twin-Fock states that is optimal for $F_E(\mathcal{I}_{A\to B}, \mathcal{T}_{A\to B}^\xi)$ is also optimal for $F_{E,M}(\mathcal{I}_{A\to B}, \mathcal{T}_{A\to B}^\xi)$ for large $M$. We know by Lemma~\ref{appendix_B_lemma3} that $F_{E,M}(\mathcal{I}_{A\to B}, \mathcal{T}_{A\to B}^\xi)$ has the same unique optimal state for large $M$. This implies the uniqueness of the optimal state \eqref{new7_eqn1} in the given sense.
\end{proof}
%Of course, the coefficients in the l.h.s of relation \ref{re:coefficiantsgeneral} are increasing with respect to $n$ for a fixed $k$.
%\begin{figure}[H]
%	\centering
%	\includegraphics[width=0.4\textwidth]{tel-pk=4-left}
%		\caption{The l.h.s of \ref{re:coefficiantsgeneral} for $k=4$, $0\leq \eta \leq 1$ and $0\leq E\leq 1$}
%	\label{fig5}
%\end{figure}

\section{Optimal state for the bidirectional teleportation protocol}

\label{app:bidirectional}

The proof of the following result is based on the ideas of \cite[Appendix B]{SSW20} and \cite[Proposition 2]{sharma2020characterizing}.
\begin{proposition} \label{new_prop2}
The energy-constrained channel fidelity \eqref{re:bidirectional} and its truncated version \eqref{new_eqn30} satisfy the inequalities
\begin{multline}\label{new10_eqn1}
    1- \left[2 \sqrt{1- \left(1-\frac{2E}{M+1} \right)^2} + \sqrt{1- F_{E,M}(\mathcal{S}_{AB \to A'B'}, \mathcal{T}_{A\to B'}^\xi\otimes \mathcal{T}_{B\to A'}^{\xi'})} \right]^2 \\ \leq F_{E}(\mathcal{S}_{AB \to A'B'}, \mathcal{T}_{A\to B'}^\xi\otimes \mathcal{T}_{B\to A'}^{\xi'}) 
    \leq F_{E,M}(\mathcal{S}_{AB \to A'B'}, \mathcal{T}_{A\to B'}^\xi\otimes \mathcal{T}_{B\to A'}^{\xi'}).
\end{multline}
\end{proposition}
\begin{proof}
By definition, we have
\begin{align}
    F_{E}(\mathcal{S}_{AB \to A'B'}, \mathcal{T}_{A\to B'}^\xi\otimes \mathcal{T}_{B\to A'}^{\xi'}) \leq F_{E,M}(\mathcal{S}_{AB \to A'B'}, \mathcal{T}_{A\to B'}^\xi\otimes \mathcal{T}_{B\to A'}^{\xi'}).
\end{align}
We now establish the inequality in the other direction.
Let $\Pi_{AB}^M$ be the $(M+1)^2$-dimensional projection operator defined as
\begin{equation}
    \Pi_{AB}^M \coloneqq \sum_{m,n=0}^{M} \vert m,n \rangle\!\langle m,n \vert_{AB}.
\end{equation}
Let $\psi_{RAB}$ be an arbitrary pure state in \eqref{new7_eqn2}. We then have
\begin{align}
    \Tr{\left(\Pi_{AB}^M \psi_{RAB} \right)} &= 1- \sum_{\substack{m,n=0 \\ \max\{m,n\} \geq M+1}}^{\infty} \langle m,n \vert \psi_{AB} \vert m,n \rangle  \\
    &= 1- \sum_{\substack{m,n=0 \\ \max\{m,n\} \geq M+1}}^{\infty} \lambda_{m,n}^2  \\
    &\geq 1- \sum_{\substack{m,n=0 \\ \max\{m,n\} \geq M+1}}^{\infty}\left( \frac{m+n}{M+1}\right) \lambda_{m,n}^2  \label{new13_eqn1}\\
    &\geq 1- \frac{2E}{M+1} \label{new_eqn31}.
\end{align}
\eqref{new13_eqn1} follows from the fact $(m+n)/(M+1) \geq 1,$ and \eqref{new_eqn31} is a consequence of the constraint $\sum_{m,n=0}^{\infty} (m+n) \lambda_{m,n}^2 \leq 2E$.
Let $\psi^M_{RAB}$ be the truncated state given by
\begin{equation}
    \psi^M_{RAB} \coloneqq \dfrac{\Pi_{AB}^M \psi_{RAB} \Pi_{AB}^M}{\Tr{\left(\Pi_{AB}^M \psi_{RAB} \right)}}.
\end{equation}
We have
\begin{align}
    F(\psi_{RAB}, \psi_{RAB}^M) & \geq \left(1-\frac{1}{2} \|\psi_{RAB}- \psi_{RAB}^M \|_1 \right)^2 \label{new13_eqn2}  \\
    & \geq \left( 1- \sqrt{\frac{2E}{M+1}} \right)^2 \label{new_eqn32} .
\end{align}
The inequality \eqref{new13_eqn2} follows from \eqref{new_eqn50}, and \eqref{new_eqn32} follows from the gentle measurement lemma (see, e.g., \cite[Lemma 9.4.1]{wilde2011classical}).
%Let $C_E(\mathcal{S}_{AB \to A'B'},\mathcal{T}_{A\to B'}^\xi\otimes \mathcal{T}_{B\to A'}^{\xi'})$ be the energy-constrained sine distance between two channels $\mathcal{N}$ and $\mathcal{M}$ given by
%\begin{align*}
%    C_E(\mathcal{S}_{AB \to A'B'},\mathcal{T}_{A\to B'}^\xi\otimes \mathcal{T}_{B\to A'}^{\xi'}) = \sup_{\rho_{RAB}: \Tr \hat{n}_{AB} \rho_{AB} \leq 2E} C_E(\mathcal{S}_{AB \to A'B'}(\rho_{RAB}),\mathcal{T}_{A\to B'}^\xi\otimes \mathcal{T}_{B\to A'}^{\xi'}(\rho_{RAB}))
%\end{align*}
%where $C(\sigma,\rho)=\sqrt{1-F(\sigma, \rho)}$ denotes the sine distance between two states $\sigma, \rho$.  
We have 
\begin{align}
   & C(\mathcal{S}_{AB \to A'B'}(\psi_{RAB}), \mathcal{T}_{A\to B'}^\xi\otimes \mathcal{T}_{B\to A'}^{\xi'}(\psi_{RAB}))\nonumber \\
    & \leq C(\mathcal{S}_{AB \to A'B'}(\psi_{RAB}), \mathcal{S}_{AB \to A'B'}(\psi_{RAB}^M)) + C(\mathcal{S}_{AB \to A'B'}(\psi_{RAB}^M), \mathcal{T}_{A\to B'}^\xi\otimes \mathcal{T}_{B\to A'}^{\xi'}(\psi_{RAB}^M)) \nonumber \\
    &\hspace{0.5cm} + C(\mathcal{T}_{A\to B'}^\xi\otimes \mathcal{T}_{B\to A'}^{\xi'}(\psi_{RAB}^M), \mathcal{T}_{A\to B'}^\xi\otimes \mathcal{T}_{B\to A'}^{\xi'}(\psi_{RAB})) \label{new13_eqn3} \\
    &\leq 2 C(\psi_{RAB}, \psi_{RAB}^M)+ C(\mathcal{S}_{AB \to A'B'}(\psi_{RAB}^M), \mathcal{T}_{A\to B'}^\xi\otimes \mathcal{T}_{B\to A'}^{\xi'}(\psi_{RAB}^M)) \label{new13_eqn4}\\
    &= 2\sqrt{1-F(\psi_{RAB}, \psi_{RAB}^M)}+\sqrt{1-F(\mathcal{S}_{AB \to A'B'}(\psi_{RAB}^M), \mathcal{T}_{A\to B'}^\xi\otimes \mathcal{T}_{B\to A'}^{\xi'}(\psi_{RAB}^M))}\\
    &\leq 2 \sqrt{1- \left( 1- \frac{2E}{M+1} \right)^2} + \sqrt{1-F_{E,M}(\mathcal{S}_{AB \to A'B'}, \mathcal{T}_{A\to B'}^\xi\otimes \mathcal{T}_{B\to A'}^{\xi'})} \label{new13_eqn5}.
\end{align}
The inequalities \eqref{new13_eqn3} and \eqref{new13_eqn4} are consequences of the triangle inequality and monotonicity of the sine distance, respectively; \eqref{new13_eqn5} follows from \eqref{new_eqn32} and the fact that
 $\psi_{RAB}^M$ is a legitimate finite dimensional state that satisfies the energy constraint $\Tr{\left( \hat{n}_{AB}\psi_{AB}^M \right)} \leq 2E$. 
The inequality \eqref{new13_eqn5} is true for arbitrary $\psi_{RAB}$ in \eqref{new7_eqn2}. So, we get
\begin{align}
    C_E(\mathcal{S}_{AB \to A'B'}, \mathcal{T}_{A\to B'}^\xi\otimes \mathcal{T}_{B\to A'}^{\xi'}) & \leq 2 \sqrt{1- \left( 1- \frac{2E}{M+1} \right)^2} + \sqrt{1-F_{E,M}(\mathcal{S}_{AB \to A'B'}, \mathcal{T}_{A\to B'}^\xi\otimes \mathcal{T}_{B\to A'}^{\xi'})}\label{new_eqn37}.
\end{align}
By squaring and then rearranging \eqref{new_eqn37} gives the desired inequality.
\end{proof}

%======================================

%\begin{proposition}
%If $E \leq \frac{1}{4 \eta \eta' + 2},$ a solution of the following optimization %problem
%$$\begin{array}{rccl}
%& \underset{p}{\operatorname{inf}}
%& & f_M(p) \\
%& \operatorname{subject \ to}
%& &\displaystyle \sum_{ m,n=0}^{M} (m+n)p_{m,n} = 2E \\
%& & & \displaystyle \sum_{ m,n=0}^{M} p_{m,n}=1 \\
%&&& p_{m,n} \geq 0 \ \ \forall m,n
%\end{array}$$
%is given by $p_{0,0}=1-2E, \ p_{0,1}=p_{1,0}=E,$ and $p_{r,s}=0$ for $r+s \geq 2$.
%The optimal value is $\eta \eta'((1-2E)^2 + 2E^2 \eta \eta')$.
%\end{proposition}

%=================================
\bigskip 

Define $f_{M,\xi, \xi'}: \mathbb{R}^{(M+1)^2} \to \mathbb{R}$ as
\begin{align}
 f_{M,\xi,\xi'}(p)&\coloneqq  \dfrac{1}{(1+\xi)(1+\xi')}\sum_{m, n, m', n'=0}^{M}p_{m,n}p_{m', n'} T_{\xi}^{mm'}  T_{\xi'}^{nn'}.
 \end{align}
We recall from \eqref{new_eqn47} that for any {\it probability vector} $p$  in $\mathbb{R}^{(M+1)^2},$ a pure state of the form $\vert \phi \rangle_{RAB} =\sum_{m,n=0}^{M} \sqrt{p_{m,n}} \vert m,n \rangle_R \vert m,n \rangle_{AB}$ 
satisfies
\begin{equation}\label{new2_eqn11}
    F(\mathcal{S}_{AB \to A'B'}(\phi_{RAB}), \mathcal{T}_{A\to B'}^\xi\otimes \mathcal{T}_{B\to A'}^{\xi'}(\phi_{RAB})) = f_{M,\xi,\xi'}(p).
\end{equation}
 By \eqref{new_eqn47} and \eqref{new_eqn30} we thus have
\begin{equation}\label{new_eqn48}
    F_{E,M}(\mathcal{S}_{AB \to A'B'}, \mathcal{T}_{A\to B'}^\xi\otimes \mathcal{T}_{B\to A'}^{\xi'}) = \inf_{p} 
   f_{M,\xi, \xi'}(p),
\end{equation}
where the minimum is taken over probability vectors $p$ in $\mathbb{R}^{(M+1)^2}$ satisfying $\sum_{m,n=0}^{M} (m+n) p_{m,n} \leq 2E$.
Let
\begin{equation}
D_M\coloneqq \{(m,n,m',n') \in \mathbb{Z}^{4}_{+}:  m,n,m',n' \leq M, m'+n' \geq 2, m+n \geq 2\},  
\end{equation}
where $\mathbb{Z}^{4}_+ \subset \mathbb{Z}^4$ is the set of $4$-tuples of non-negative integers.
We can rewrite $f_{M, \xi, \xi'}$ as
\begin{align}
(1+\xi)(1+\xi') f_{M,\xi, \xi'}(p)&= \sum_{\substack{m,n,m',n'=0 \\ (m,n,m',n') \notin D_M}}^M  p_{m,n}p_{m',n'} T_{\xi}^{mm'}  T_{\xi'}^{nn'} +\sum_{\substack{m,n,m',n'=0 \\ (m,n,m',n') \in D_M}}^M p_{m,n}p_{m', n'} T_{\xi}^{mm'}  T_{\xi'}^{nn'} \nonumber \\\nonumber\\
 &= \sum_{\substack{m,n,m',n'=0 \\ m+n \leq 1 \\ m'+n'\leq 1}}^M p_{m,n}p_{m',n'} T_{\xi}^{mm'}  T_{\xi'}^{nn'} 
 +  \sum_{\substack{m,n,m',n'=0 \\m+n \leq 1\\ m'+n' \geq 2}}^M p_{m,n}p_{m',n'} T_{\xi}^{mm'}  T_{\xi'}^{nn'}  \nonumber \\
 &\hspace{0.5cm}+ \sum_{\substack{m,n,m',n'=0 \\ m+n \geq 2\\ m'+n' \leq 1}}^M p_{m,n}p_{m',n'} T_{\xi}^{mm'}  T_{\xi'}^{nn'}  + \sum_{\substack{m,n,m',n'=0 \\ (m,n,m',n') \in D_M}}^M p_{m,n}p_{m',n'} T_{\xi}^{mm'}  T_{\xi'}^{nn'}. \label{new3_eqn2}
 \end{align}
%We now simplify the first two terms in \eqref{re:bidirectionalfunction}.
From \eqref{new3_eqn1}, using $\eta=1/(1+\xi)$ we get
\begin{equation}\label{re:bidirectionaltrace}
  \mathcal{L}^{\frac{1}{1+\xi}}(\vert m\rangle\!\langle m'\vert)\mathcal{L}^{\frac{1}{1+\xi}}(\vert m'\rangle\!\langle m\vert)=
 \sum_{k=0}^{\min\{m, m'\}}   \binom{m}{k} \binom{m'}{k} \dfrac{\xi^{2k}}{(1+\xi)^{m+m'}}  \vert m-k\rangle\!\langle m-k \vert.
\end{equation}
By taking the trace on both sides of \eqref{re:bidirectionaltrace}, and from \eqref{new_eqn38a} we get
\begin{align}  
T_{\xi}^{mm'} 
&=\sum_{k=0}^{\min\{m, m'\}}\binom{m}{k} \binom{m'}{k} \dfrac{\xi^{2k}}{(1+\xi)^{m+m'}}.\label{re:biditrace(a)}
\end{align}
In particular, $T^{0k}_{\xi}= 1/(1+\xi)^k$ and $T^{1k}_{\xi} =  1/(1+\xi)^{k+1} +k \xi^2/(1+\xi)^{k+1}$ for all $k \geq 0$.
Similarly 
\begin{align}  
T_{\xi'}^{nn'}
&=\sum_{k=0}^{\min\{n, n'\}}\binom{n}{k} \binom{n'}{k} \dfrac{\xi'^{2k}}{(1+\xi')^{n+n'}}.\label{re:biditrace(b)}
\end{align}
We also note that  $T^{mm'}_{\xi}=T^{m'm}_{\xi}$ and  $T^{nn'}_{\xi'}=T^{n'n}_{\xi'}$. From \eqref{new3_eqn2} we thus get
\begin{multline}
 f_{M,\xi, \xi'}(p) =\dfrac{1}{(1+\xi)(1+\xi')} \left[ \sum_{\substack{m,n,m',n'=0 \\ m+n \leq 1 \\ m'+n'\leq 1}}^M p_{m,n}p_{m',n'} T_{\xi}^{mm'}  T_{\xi'}^{nn'}  
 + 2\sum_{\substack{m,n,m',n'=0 \\m+n \leq 1\\ m'+n' \geq 2}}^M p_{m,n}p_{m',n'} T_{\xi}^{mm'}  T_{\xi'}^{nn'}  \right. \\ \left.+ \sum_{\substack{m,n,m',n'=0 \\ (m,n,m',n') \in D_M}}^M  p_{m,n}p_{m', n'} T_{\xi}^{mm'}  T_{\xi'}^{nn'}  \right].\label{re:bidirectionalfunction}
 \end{multline}

\subsection{Optimal input state for the energy-constrained channel fidelity for $\xi = \xi'$}

\label{app:bidir-opt-state-xi'-equal-xi}

\begin{proposition}\label{new_prop1}
The minimum of $f_{M,\xi,\xi},$ subject to $p_{m,n} \geq 0$ for all $m,n\in\{0,\ldots,M\},$
\begin{equation}\label{new_eqn45}
    \sum_{m,n=0}^{M} p_{m,n}=1, \quad  \sum_{m,n=0}^{M} (m+n) p_{m,n} \leq 2E,
\end{equation}
is attained at the unique point $p=(p_{m,n})_{m,n=0}^M$ in $\mathbb{R}^{(M+1)^2},$ given by
\begin{equation}\label{new12_eqn3}
    p_{m,n}= \begin{cases}
       1-2E & \text{ if } m=n=0, \\
       E & \text{ if } m+n=1, \\
       0 & \text{ if } m+n \geq 2,
    \end{cases}
\end{equation}
whenever $2E \leq (1+\xi)/(2+3\xi)$. 
\end{proposition}
\begin{proof}
We divide the proof into two parts, similar to the proof of Proposition~\ref{new3_prop1}. Let $E_0 \in [0,E]$. In the first part, we show that the minimum of $f_{M,\xi,\xi},$ subject to $ p_{m,n} \geq 0$ for all $m,n\in\{0,\ldots,M\},$ and the equality constraints
\begin{equation}\label{new12_eqn6}
    \sum_{m,n=0}^{M} p_{m,n}=1, \quad  \sum_{m,n=0}^{M} (m+n) p_{m,n} = 2E_0,
\end{equation}
is uniquely attained at $p$ as given in \eqref{new12_eqn3}. In the second part, we show that the minimum value obtained in the first part is a strictly decreasing function of $E_0$. It then follows that the minimizer of $f_{M,\xi,\xi}$ subject to \eqref{new_eqn45} is given by \eqref{new12_eqn3}.

\emph{Part 1.}
In \eqref{re:bidirectionalfunction}, substitute the values of $T_{\xi}^{mm'}$ and $T_{\xi}^{nn'}$ from \eqref{re:biditrace(a)} and \eqref{re:biditrace(b)} to get
 \begin{multline}\label{new_eqn14}
 f_{M,\xi,\xi}(p) =\eta^2 \Bigg[ p_{0,0}^2+ 2\eta p_{0,0}p_{0,1} + 2\eta p_{0,0}p_{1,0} + 2\eta^2 p_{0,1}p_{1,0}+(\eta^2 + (1-\eta)^2)(p_{0,1}^2  + p_{1,0}^2)   \\
 + 2 \sum_{\substack{m,n=0 \\ m+n \geq 2}}^M \left( \eta^{m+n}  p_{0,0}  +  \eta^m (\eta^{n+1}+n \eta^{n-1} (1-\eta)^2) p_{0,1}   + \eta^n (\eta^{m+1}+m \eta^{m-1} (1-\eta)^2) p_{1,0}   \right)p_{m,n}   \\
+  \sum_{\substack{m,n,m',n'=0 \\ (m,n,m',n') \in D_M}}^M p_{m,n}p_{m', n'}T_{\xi}^{mm'}  T_{\xi}^{nn'} \Bigg]. 
 \end{multline}
 Here we used the relation $\eta=1/(1+\xi)$ for making the following calculations convenient.
%\begin{align}
% f_M(p)&= \eta \eta' (p_{0,0}^2+p_{0,1}^2 (1-\eta) (1-\eta') + p_{1,0}^2 (1-\eta) (1-\eta')) + 2 \eta \eta' \sum_{m=1}^{M} p_{0,0}p_{m,m} (\eta\eta')^m \nonumber \\
% & \hspace{0.5cm} +2 \eta \eta' (1-\eta)(1-\eta')  \sum_{m=2}^{M} m \left(p_{0,1}p_{m-1,m}+p_{1,0} p_{m,m-1} \right) (\eta\eta')^{m-1}\nonumber \\
% & \hspace{2cm} +\eta \eta' \sum_{(m,n,m',n') \in D_M}p_{m,n}p_{m', n'} T_{\eta \eta'}^{mnm'n'}. \label{new_eqn14}
% \end{align}
%Since the  optimization of $f(p)$ is a convex and increasing function respect to $E$ then $f(p)$ takes its maximum in the boundary of energy constrained condition.
% The other acceptable condition is that $p_{0,0}$ and $p_{1,0}$ can be considered as functions in the following manner
From \eqref{new12_eqn6}, we can write $p_{0,0}$ and $p_{0,1}$ as
\begin{align}
  p_{0,1}&=2E_0-p_{1,0}-\sum_{\substack{m,n=0 \\ m+n \geq 2}}^{M} (m+n) p_{m,n}, \label{re:energycondi}\\
  p_{0,0}&=1-2E_0+\sum_{\substack{m,n=0 \\ m+n \geq 2}}^{M} (m+n-1) p_{m,n} \label{new_eqn15}.
\end{align}
Using the relations \eqref{re:energycondi} and \eqref{new_eqn15}, we can treat $p_{0,0},p_{0,1}$ as dependent variables so that
$f_{M,\xi,\xi}(p)$ is a function of $(M+1)^2-2$ independent variables
$\{p_{r,s}: r,s\in \{0,\ldots, M\}\}\backslash \{(0,0),(0,1)\}$.
We have
\begin{equation}
 \frac{\partial p_{0,0}}{\partial p_{r,s}} =r+s-1,\quad  \frac{\partial p_{0,1}}{\partial p_{r,s}}=-(r+s).
\end{equation}
We now argue that a necessary condition for a minimizer is $p_{r,s}=0$ whenever $r+s \geq 2$.
Differentiate $f_{M,\xi,\xi}$ partially with respect to $p_{r,s}$ for $r+s \geq 2$.
We get
\begin{multline}%\label{new_eqn33}
\dfrac{1}{2\eta^2} \pdv{f_{M,\xi,\xi}(p)}{p_{r,s}} = (r+s-1)p_{0,0}+\eta \left[-(r+s)p_{0,0}+(r+s-1)p_{0,1}\right]+\eta \left[ (r+s-1)- \eta(r+s) \right]p_{1,0}  \\
 -\left[\eta^2+(1-\eta)^2\right](r+s)p_{0,1} +  \eta^{r+s} p_{0,0} + \eta^r \left[\eta^{s+1}+s \eta^{s-1} (1-\eta)^2\right] p_{0,1}   +  \eta^s \left[\eta^{r+1}+r \eta^{r-1} (1-\eta)^2 \right]  p_{1,0}   \\
 + \sum_{\substack{m,n=0 \\ m+n \geq 2}}^M \left[ (r+s-1) \eta^{m+n} - (r+s) \eta^m (\eta^{n+1}+n \eta^{n-1} (1-\eta)^2) + T^{mr}_{\xi}T^{ns}_{\xi}  \right] p_{m,n}.
 \end{multline}
Further simplification gives
\begin{align}\label{new_eqn34}
\dfrac{1}{2\eta^2} \pdv{f_{M,\xi,\xi}(p)}{p_{r,s}} &= \left[(r+s-1)-\eta(r+s)+\eta^{r+s} \right] p_{0,0}\nonumber\\
&\hspace{0.5cm}+ \left[\eta (r+s-1)-(\eta^2+(1-\eta)^2)(r+s)+\eta^{r+s+1}+s \eta^{r+s-1}(1-\eta)^2 \right]p_{0,1}  \nonumber \\
&\hspace{0.5cm} + \left[\eta( (r+s-1)- \eta(r+s))+\eta^{r+s+1}+r \eta^{r+s-1}(1-\eta)^2 \right]p_{1,0} \nonumber \\
&\hspace{0.5cm} + \sum_{\substack{m,n=0 \\ m+n \geq 2}}^M \left[ (r+s-1) \eta^{m+n} - (r+s)(  \eta^{m+n+1}+n \eta^{m+n-1}(1-\eta)^2) + T^{mr}_{\xi}T^{ns}_{\xi}  \right] p_{m,n}.
 \end{align}
Substitute the values of $p_{0,0}$ and $p_{0,1}$ in \eqref{new_eqn34} from \eqref{re:energycondi} and \eqref{new_eqn15},  and simplify to get
\begin{align}\label{new_eqn35}
\dfrac{1}{2\eta^2} \pdv{f_{M,\xi,\xi}(p)}{p_{r,s}}
&=(r+s)(1-\eta)-(1-\eta^{r+s}) -2E_0 (1-\eta)  \left[(2(r+s)-s\eta^{r+s-1})(1 - \eta)-(1-\eta^{r+s})\right]\nonumber \\  
&\hspace{0.5cm} + (1-\eta)^2 \left[(r+s)+(r-s)\eta^{r+s-1} \right]p_{1,0}
 + \sum_{\substack{m,n=0 \\ m+n \geq 2}}^M \Gamma_{m,n,r,s,\eta} p_{m,n},
 \end{align}
where
\begin{align}\label{coeeqn}
    \Gamma_{m,n,r,s,\eta} &= (1-\eta)^2\left[2(m+n)(r+s)-\dfrac{1-\eta^{r+s}}{1-\eta} (m+n)-\dfrac{1-\eta^{m+n}}{1-\eta} (r+s) \right]+1 - \eta^{m+n}-\eta^{r+s}+ T^{mr}_{\eta}T^{ns}_{\eta}\nonumber \\
    &\hspace{0.5cm} -(1-\eta)^2\left[s(m+n)\eta^{r+s-1} +n(r+s)\eta^{m+n-1} \right].
\end{align}
The following arguments show that $\Gamma_{m,n,r,s,\eta}>0.$ We have
\begin{align}
    T^{mr}_{\eta} &\geq \eta^{m+r},\label{coeqn2} \\
    T^{ns}_{\eta} &\geq \eta^{n+s} \left[1+ ns\left( \dfrac{1-\eta}{\eta} \right)^2 \right]\label{coeqn3} .
\end{align}
This implies
\begin{align}
     T^{mr}_{\eta}  T^{ns}_{\eta} &\geq \eta^{m+n} \eta^{r+s} \left[1+ [(r+s)-r][(m+n)-m]\left( \dfrac{1-\eta}{\eta} \right)^2 \right].\label{coeqn5}
\end{align}
Let $\alpha=r+s$ and $\beta=m+n.$ Using the inequality \eqref{coeqn5}, from \eqref{coeeqn} we get
\begin{align}
     \Gamma_{m,n,r,s,\eta}& \geq   (1-\eta)^2\left[2\alpha \beta-\dfrac{1-\eta^{\alpha}}{1-\eta} \beta-\dfrac{1-\eta^{\beta}}{1-\eta} \alpha \right]+1 - \eta^{\beta}-\eta^{\alpha} +\eta^{\alpha + \beta} \left[1+ (\alpha-r)(\beta-m)\left( \dfrac{1-\eta}{\eta} \right)^2 \right] \nonumber \\
     &\hspace{0.5cm}-(1-\eta)^2\left[(\alpha-r)\beta \eta^{\alpha-1} +(\beta-m)\alpha \eta^{\beta-1} \right] \\
     &= (1-\eta)^2\left[2 \alpha \beta-\dfrac{1-\eta^{\alpha}}{1-\eta} \beta-\dfrac{1-\eta^{\beta}}{1-\eta} \alpha \right]+1 - \eta^{\beta}-\eta^{\alpha}+\eta^{\alpha + \beta} \nonumber  \\
     & \hspace{0.5cm}+(1-\eta)^2 \left[    (\alpha-r)(\beta-m)\eta^{\alpha+\beta-2}-(\alpha-r)\beta\eta^{\alpha-1} -(\beta-m)\alpha\eta^{\beta-1} \right] \\
     &= (1-\eta)^2\left[2\alpha \beta-\dfrac{1-\eta^{\alpha}}{1-\eta} \beta-\dfrac{1-\eta^{\beta}}{1-\eta} \alpha \right]+(1 - \eta^{\alpha})(1-\eta^{\beta}) \nonumber  \\
     &\hspace{0.5cm} +(1-\eta)^2 \left[ (\alpha \beta-\beta r-\alpha m+\beta m)\eta^{\alpha-1}\eta^{\beta-1}-(\alpha \beta -\beta r)\eta^{\alpha-1} -(\alpha \beta-\alpha m)\eta^{\beta-1} \right]\\
     &= (1-\eta)^2\left[2\alpha \beta-\dfrac{1-\eta^{\alpha}}{1-\eta} \beta-\dfrac{1-\eta^{\beta}}{1-\eta} \alpha \right]+(1 - \eta^{\alpha})(1-\eta^{\beta})+ (1-\eta)^2 \alpha \beta\left[ \eta^{\alpha-1}\eta^{\beta-1}-\eta^{\alpha-1}-\eta^{\beta-1} \right] \nonumber \\
     &\hspace{0.5cm}  +(1-\eta)^2 \left[\beta r\eta^{\alpha-1}(1-\eta^{\beta-1}) + \alpha m \eta^{\beta-1}(1-\eta^{\alpha-1})+  rm \eta^{\alpha-1}\eta^{\beta-1} \right] \\
     &\geq (1-\eta)^2\left[2\alpha \beta-\dfrac{1-\eta^{\alpha}}{1-\eta} \beta-\dfrac{1-\eta^{\beta}}{1-\eta} \alpha \right]+(1 - \eta^{\alpha})(1-\eta^{\beta})+ (1-\eta)^2 \alpha \beta\left[ \eta^{\alpha-1}\eta^{\beta-1}-\eta^{\alpha-1}-\eta^{\beta-1} \right]\\
     &= (1-\eta)^2\left[\alpha \beta (1+\eta^{\alpha-1}\eta^{\beta-1}-\eta^{\alpha-1}-\eta^{\beta-1})+\alpha \beta -\dfrac{1-\eta^{\alpha}}{1-\eta} \beta-\dfrac{1-\eta^{\beta}}{1-\eta} \alpha +\left(\dfrac{1 - \eta^{\alpha}}{1-\eta} \right)\left(\dfrac{1-\eta^{\beta}}{1-\eta}\right)\right]\\
     &= (1-\eta)^2\left[\alpha\beta(1-\eta^{\alpha-1})(1-\eta^{\beta-1})+\left(\alpha- \sum_{i=0}^{\alpha-1} \eta^i  \right) \left( \beta- \sum_{i=0}^{\beta-1} \eta^i  \right)\right] >0.
\end{align}
Now, the constant term is non-negative if
\begin{align}\label{new_eqn42}
    2E_0 (1-\eta) \left[(2\alpha-s\eta^{\alpha-1})(1 - \eta)-(1-\eta^{\alpha})\right] \leq \alpha(1-\eta)-(1-\eta^{\alpha}).
\end{align}
Using the fact that $s \geq 0,$ the inequality \eqref{new_eqn42}  holds if
\begin{align}\label{new_eqn43}
    2E_0 (1-\eta) \left[2\alpha(1 - \eta)-(1-\eta^{\alpha})\right] \leq \alpha(1-\eta)-(1-\eta^{r+s}).
\end{align}
By basic real analysis it is easy to verify that the coefficient of $E_0$ in \eqref{new_eqn43} is positive. 
So, the inequality \eqref{new_eqn43} is equivalent to
\begin{align}\label{new_eqn51}
    2E_0  &\leq \dfrac{\alpha(1-\eta)-(1-\eta^{\alpha})}{(1-\eta) \left[2\alpha(1 - \eta)-(1-\eta^{\alpha})\right]}.
\end{align}
The right-hand side expression in \eqref{new_eqn51} is an increasing function of $\alpha \geq 2$. To verify this, it suffices to show that
\begin{align}
   \dfrac{\alpha(1-\eta)-(1-\eta^\alpha)}{(1-\eta)[2\alpha(1-\eta)-(1-\eta^\alpha)]}- \dfrac{2(1-\eta)-(1-\eta^2)}{(1-\eta)[4(1-\eta)-(1-\eta^2)]} &\geq 0 \Longleftrightarrow  \\  
   \dfrac{\alpha(1-\eta)-(1-\eta^\alpha)}{2\alpha(1-\eta)-(1-\eta^\alpha)}-\dfrac{1-\eta}{3-\eta} &\geq 0 \Longleftrightarrow \\
  1- \dfrac{\alpha(1-\eta)}{2\alpha(1-\eta)-(1-\eta^\alpha)}-\dfrac{1-\eta}{3-\eta}&\geq 0 \Longleftrightarrow \\
 \dfrac{2}{3-\eta}-\dfrac{\alpha(1-\eta)}{2\alpha(1-\eta)-(1-\eta^\alpha)} &\geq 0 \Longleftrightarrow \\
   4 \alpha (1-\eta)-2(1-\eta^\alpha)-\alpha (1-\eta)(3-\eta) &\geq 0 \Longleftrightarrow  \\
      4\alpha-4\alpha\eta-2+2\eta^\alpha- [3\alpha-4\alpha\eta+\alpha\eta^2 ]&\geq 0 \Longleftrightarrow  \\
     \alpha-2+2\eta^\alpha-\alpha\eta^2 &\geq 0 \Longleftrightarrow \\
   \alpha(1-\eta^2)-2(1-\eta^\alpha) &\geq 0 \Longleftrightarrow \\
 \alpha(1+\eta)-2\left(\dfrac{1-\eta^\alpha}{1-\eta} \right) &\geq 0 \Longleftrightarrow \\
  \alpha+\alpha\eta-2-2\sum_{i=1}^{\alpha-1} \eta^i &\geq 0 \Longleftrightarrow \\
 \alpha+\alpha\eta-2 -2(\alpha-1) \eta +2 \eta\sum_{i=1}^{\alpha-1}(1- \eta^{i-1}) &\geq 0 \Longleftrightarrow  \\
 \alpha-\alpha \eta -2+2\eta +2 \eta\sum_{i=1}^{\alpha-1}(1- \eta^{i-1}) &\geq 0 \Longleftrightarrow  \\
 (\alpha-2)(1-\eta)+2\sum_{i=1}^{\alpha-1} \eta(1-\eta^{i-1}) \geq 0. \label{new15eqn1}
\end{align}
The inequality \eqref{new15eqn1} holds because $\alpha=r+s\geq 2.$
A sufficient condition on $E$ for the inequality \eqref{new_eqn43} is thus obtained by keeping $2E$ not more than the minimum value of the right hand side of \eqref{new_eqn51} which is attained for $\alpha=2.$ 
This is given by 
\begin{equation}\label{new13_eqn7}
     2E \leq \dfrac{1}{3-\eta}=\dfrac{1+\xi}{2+3\xi}.
\end{equation}
%The sufficient condition \eqref{new13_eqn7} holds by hypothesis.
% \begin{figure}
% \begin{subfigure}{.5\textwidth}
%   \centering
%   \includegraphics[width=.8\linewidth]{figures/constant_0.2.pdf}
%   \caption{$\eta=0.2$}
%   \label{constant_0.2}
% \end{subfigure}%
% \begin{subfigure}{.5\textwidth}
%   \centering
%   \includegraphics[width=.8\linewidth]{figures/constant_0.4.pdf}
%   \caption{$\eta=0.4$}
%   \label{constant_0.4}
% \end{subfigure}
% \begin{subfigure}{.5\textwidth}
%   \centering
%   \includegraphics[width=.8\linewidth]{figures/constant_0.6.pdf}
%   \caption{$\eta=0.6$}
%   \label{constant_0.6}
% \end{subfigure}%
% \begin{subfigure}{.5\textwidth}
%   \centering
%   \includegraphics[width=.8\linewidth]{figures/constant_0.8.pdf}
%   \caption{$\eta=0.8$}
%   \label{constant_0.8}
% \end{subfigure}
% \caption{The right-hand side of \eqref{new_eqn51} for $2 \leq r+s \leq 10$}
% \label{plot_constant_term}
% \end{figure}

We have thus shown that $\partial f_{M,\xi,\xi}(p)/\partial p_{r,s} > 0,$ whenever $p_{r,s} > 0$ and $r+s \geq 2$.
So, $f_{M,\xi, \xi}$ is a strictly increasing function of the variables $p_{r,s}$
such that $r+s \geq 2$.
Let $q$ be a minimizer of $f_{M,\xi,\xi},$ which exists because $f_{M,\xi,\xi}$ is a continuous function over the compact set \eqref{new12_eqn6}.
From the necessary condition derived earlier, we must have $q_{r,s}=0$ for all $r+s \geq 2$. 
Thus, we get 
\begin{equation}
    q_{0,0}=1-2E_0, \quad q_{0,1}=2E_0-q_{1,0},
\end{equation}
from \eqref{re:energycondi} and \eqref{new_eqn15}.
This gives 
\begin{align}\label{neq_eqn56}
f_{M,\xi,\xi}(q)= \eta^2 \bigg[1-4(1-\eta)E_0+4(1-4\eta+2\eta^2)E_0^2 - 4(1-\eta)^2E_0 q_{1,0} + 2(1-\eta)^2 q_{1,0}^2 \bigg]. 
 \end{align}
This is a convex polynomial in $q_{1,0},$
which has the unique minimizer $q_{1,0}=E_0$. In other words, the minimizer of $f_{M,\xi,\xi},$ subject to the constraints \eqref{new12_eqn6}, is given by $q_{0,0}=1-2E_0,$ $q_{0,1}=q_{1,0}=E_0,$ and $q_{m,n}=0$ for all $m,n$ with $m+n \geq 2$.

\emph{Part 2.}
By evaluating  \eqref{neq_eqn56} at the minimizer $q$ obtained in the first part, and re-substituting $\eta=1/(1+\xi),$ we get 
\begin{align}\label{new12_eqn5}
    f_{M,\xi,\xi}(q)=\dfrac{1}{(1+\xi)^2}\left[1-4 \left(\dfrac{\xi}{1+\xi} \right)E_0+6\left(\dfrac{\xi}{1+\xi} \right)^2E_0^2   \right],
\end{align}
which is a strictly decreasing function of $E_0$ in the interval 
$\left[0, (1+\xi)/(3\xi)\right]$. Also, we have $E_0 \leq E \in \left[0, (1+\xi)/(3\xi)\right],$ which follows from the hypothesis $E \leq (1+\xi)/(2+3\xi) < (1+\xi)/(3\xi)$. This completes the proof.
\end{proof}

%The truncated minimization problem \eqref{new_eqn30} is equivalent to the following finite dimensional counterpart of the infinite dimensional optimization problem \eqref{re:bidirectional}.

%\begin{equation}\label{new_eqn23}
%F_{E,M}(\mathcal{S}_{AB \to A'B'}, \mathcal{T}_{A\to B'}^\xi\otimes \mathcal{T}_{B\to A'}^{\xi'})=
%\begin{cases}
%   \begin{aligned}
%& \underset{p}{\text{inf}}
%& & f_M(p) \\
%& \text{subject to}
%& & \sum_{ m,n=0}^{M} (m+n)p_{m,n} \leq 2E, \\
%&&& \sum_{ m,n=0}^{M} p_{m,n}=1, \\
%& & & p_{m,n} \geq 0 \ \ \forall \ 0 \leq m,n \leq M.
%\end{aligned}
%\end{cases}
%\end{equation}

%We now find a solution to the truncated problem \eqref{new_eqn30} and then use the inequalities in Proposition~\ref{new_prop2} to find
%a solution to the general problem \eqref{new_eqn1} whenever $\xi=\xi'$.

\begin{lemma}\label{new_lem2}
The $M$-truncated energy-constrained channel fidelity \eqref{new_eqn30} has the unique optimal state given by \eqref{new5_eqn1}, whenever $\xi = \xi'$ and $2E \leq (1+\xi)/(2+3\xi)$.
Furthermore, we have
\begin{align}\label{new_eqn26}
    F_{E,M}(\mathcal{S}_{AB \to A'B'}, \mathcal{T}_{A\to B'}^\xi\otimes \mathcal{T}_{B\to A'}^{\xi}) &= \dfrac{1}{(1+\xi)^2}\left[1-4 \left(\dfrac{\xi}{1+\xi} \right)E+6\left(\dfrac{\xi}{1+\xi} \right)^2E^2   \right].
\end{align}
In particular, the the right hand side of \eqref{new_eqn26} is independent of the truncation parameter $M$.
\end{lemma}
\begin{proof}
Let $\psi_{RAB}$ be a pure state given by $\vert \psi \rangle_{RAB}=  \sum_{m,n=0}^{M} \sqrt{p_{m,n}} \vert m,n\rangle_R \vert m,n\rangle_{AB},$ where $p=(p_{m,n})_{m,n=0}^M$ is a probability vector such that $\sum_{m,n=0}^{M} (m+n) p_{m,n} \leq 2E$.
From \eqref{new2_eqn11}, we have
\begin{align}\label{new_eqn25}
    F(\mathcal{S}_{AB \to A'B'}(\phi_{RAB}), \mathcal{T}^{\xi}_{A \to B'} \otimes \mathcal{T}^{\xi}_{B \to A'} (\phi_{RAB}))= f_{M,\eta, \eta}(p).
\end{align}
By Proposition~\ref{new_prop1}, for $p$ in \eqref{new12_eqn3}, we have
\begin{align}\label{new_eqn49}
     F(\mathcal{S}_{AB \to A'B'}(\psi_{RAB}), \mathcal{T}_{A\to B'}^\xi\otimes \mathcal{T}_{B\to A'}^{\xi}(\psi_{RAB})) & \geq f_{M,\xi,\xi}(p) = \dfrac{1}{(1+\xi)^2}\left[1-4 \left(\dfrac{\xi}{1+\xi} \right)E+6\left(\dfrac{\xi}{1+\xi} \right)^2E^2   \right].
\end{align}
By taking infimum in \eqref{new_eqn49} over $\psi_{RAB},$ we get
\begin{align}\label{new12_eqn26}
    F_{E,M}(\mathcal{S}_{AB \to A'B'}, \mathcal{T}_{A\to B'}^\xi\otimes \mathcal{T}_{B\to A'}^{\xi}) &\geq \dfrac{1}{(1+\xi)^2}\left[1-4 \left(\dfrac{\xi}{1+\xi} \right)E+6\left(\dfrac{\xi}{1+\xi} \right)^2E^2   \right].
\end{align}
The inequality in \eqref{new_eqn49} is saturated for the state \eqref{new5_eqn1}, which corresponds to the minimizer of $f_{M,\xi,\xi}$ by Proposition~\ref{new_prop1}.
Therefore, the inequality \eqref{new12_eqn26} is actually an equality.
The uniqueness of the optimal state follows from Proposition~\ref{new_prop1}.
\end{proof}

\begin{theorem}\label{main_thm2}
The energy-constrained channel fidelity \eqref{re:bidirectional} has an optimal input state given by \eqref{new5_eqn1},
whenever $\xi=\xi'$ and $2E \leq (1+\xi)/(2+3\xi)$. Moreover, the optimal state is unique in the sense that there is no other optimal finite entangled superposition of twin-Fock states. The energy-constrained channel fidelity is given by 
\begin{equation} \label{new10_eqn2}
    F_E(\mathcal{S}_{AB \to A'B'}, \mathcal{T}_{A\to B'}^\xi\otimes \mathcal{T}_{B\to A'}^{\xi})=\dfrac{1}{(1+\xi)^2}\left[1-4 \left(\dfrac{\xi}{1+\xi} \right)E+6\left(\dfrac{\xi}{1+\xi} \right)^2E^2   \right].
\end{equation}

\end{theorem}
\begin{proof}
The equality  \eqref{new10_eqn2} follows directly by substituting the value of $F_{E,M}(\mathcal{S}_{AB \to A'B'}, \mathcal{T}_{A\to B'}^\xi\otimes \mathcal{T}_{B\to A'}^{\xi})$ in \eqref{new10_eqn1}, and then by taking limit the limit $M \to \infty$.
It thus follows by Lemma~\ref{new_lem2} that the state in \eqref{new5_eqn1} is an optimal state.

Any optimal finite entangled superposition of twin-Fock states for $F_{E}(\mathcal{S}_{AB \to A'B'}, \mathcal{T}_{A\to B'}^\xi\otimes \mathcal{T}_{B\to A'}^{\xi})$ is also optimal for $F_{E,M}(\mathcal{S}_{AB \to A'B'}, \mathcal{T}_{A\to B'}^\xi\otimes \mathcal{T}_{B\to A'}^{\xi})$ for large $M$. Moreover, for all $M,$ $F_{E,M}(\mathcal{S}_{AB \to A'B'}, \mathcal{T}_{A\to B'}^\xi\otimes \mathcal{T}_{B\to A'}^{\xi})$ has the same unique optimal state given by \eqref{new5_eqn1}, which follows from Lemma~\ref{new_lem2}. This implies the uniqueness of the optimal state for $F_{E}(\mathcal{S}_{AB \to A'B'}, \mathcal{T}_{A\to B'}^\xi\otimes \mathcal{T}_{B\to A'}^{\xi}),$ in the given sense.
\end{proof}

\subsection{Optimal input state for energy-constrained channel fidelity for $\xi' \geq 1$ and arbitrary $\xi$}

\label{app:bidir-opt-state-xi'-xi}

\begin{proposition}\label{new2_prop1}
The minimum of $f_{M,\eta,\eta'},$ subject to  $p_{m,n} \geq 0$ for all $m,n\in \{0,\ldots, M\},$ and the equality constraints
\begin{equation}\label{new2_eqn45}
    \sum_{m,n=0}^{M} p_{m,n}=1, \quad  \sum_{m,n=0}^{M} (m+n) p_{m,n} = 2E,
\end{equation}
is attained at the unique $p=(p_{m,n})_{m,n=0}^M$ in $\mathbb{R}^{(M+1)^2},$ given by
\begin{equation}\label{new2_eqn8}
    p_{m,n}= \begin{cases}
       1-2E & \text{ if } m=n=0, \\
       2E-p_{E} & \text{ if } m=0,n=1, \\
       0 & \text{ if } m+n \geq 2,
    \end{cases}
\end{equation}
whenever $\xi' \geq 1$ and $2E \leq (\xi'^2-1)/(\xi'(3\xi'-1))$. 
In \eqref{new2_eqn8}, we have
\begin{equation}\label{new2_eqn10}
  p_{E}=  \begin{cases}
   0 & \text{ if } 2E < \dfrac{(\xi - \xi') (1+\xi')}{\xi'^2 (1+\xi)}, \\
   2E & \text{ if } 2E < \dfrac{(\xi'-\xi)(1+\xi)(1+\xi')}{\xi^2 (1+\xi')^2+ (\xi-\xi')^2}, \\
   \dfrac{(\xi-\xi')(1+\xi)(1+\xi')+2E \xi'^2 (1+\xi)^2}{2 \left((\xi-\xi')^2+\xi \xi' (1+\xi)(1+\xi') \right)} & \text{ otherwise}.
\end{cases}
\end{equation}
\end{proposition}
\begin{proof}
From \eqref{re:bidirectionalfunction}, and using the relations $\eta=1/(1+\xi)$ and $\eta'=1/(1+\xi'),$ we get
\begin{align}\label{new2_eqn1}
 f_{M,\xi, \xi'}(p) \geq \dfrac{1}{(1+\xi)(1+\xi')} \sum_{\substack{m,n,m',n'=0 \\m+n \leq 1, \\ m'+n'\leq 1}}^M  p_{m,n}p_{m',n'} T_{\xi}^{mm'}  T_{\xi'}^{nn'}.
 \end{align}
Simplify the right-hand side of  \eqref{new2_eqn1} using $\eta=1/(1+\xi)$ and $\eta'=1/(1+\xi')$ to get
\begin{align}\label{new2_eqn2}
 f_{M,\xi, \xi'}(p) \geq \eta \eta' \left[ p_{0,0}^2+2 p_{0,0}(\eta' p_{0,1}+\eta p_{1,0})+2 \eta \eta' p_{0,1}p_{1,0}+ (\eta'^2+(1-\eta')^2)p_{0,1}^2+ (\eta^2+(1-\eta)^2)p_{1,0}^2\right].
 \end{align}
Let us denote the right-hand side of \eqref{new2_eqn2} by $g_{M,\eta,\eta'}:\mathbb{R}^{(M+1)^2} \to \mathbb{R},$
\begin{align}
    g_{M,\eta, \eta'}(p)& \coloneqq \eta \eta' \left[ p_{0,0}^2+2 p_{0,0}(\eta' p_{0,1}+\eta p_{1,0})+2 \eta \eta' p_{0,1}p_{1,0}+ (\eta'^2+(1-\eta')^2)p_{0,1}^2+ (\eta^2+(1-\eta)^2)p_{1,0}^2\right].
\end{align}
We thus have
\begin{align}\label{new2_eqn3}
 f_{M,\eta, \eta'}(p) \geq  g_{M,\eta, \eta'}(p).
 \end{align}
We will show that $g_{M,\eta,\eta'}$ has a unique minimizer, and then show that it is also the minimizer of $f_{M,\xi,\xi'}$.
Differentiate $g_{M,\eta,\eta'}$ with respect to $p_{r,s}$ for $r+s \geq 2$. We get
\begin{multline}\label{new2_eqn4}
\dfrac{1}{2 \eta \eta'}\pdv{g_{M,\eta,\eta'}(p)}{p_{r,s}} = \left((r+s)(1-\eta')-1 \right)p_{0,0}- \left((r+s)(1-2\eta')(1-\eta')+\eta' \right)p_{0,1} +\eta \left((r+s)(1-\eta')-1 \right)p_{1,0}.
 \end{multline}
Substitute the values of $p_{0,0}$ and $p_{0,1}$ from \eqref{re:energycondi} and \eqref{new_eqn15} in \eqref{new2_eqn4},  and simplify to get
\begin{align}\label{new2_eqn5}
\dfrac{1}{2 \eta \eta'}\pdv{g_{M,\eta,\eta'}(p)}{p_{r,s}} &= (r+s)(1-\eta')-1 - 2E\left((1-\eta')(2(r+s)(1-\eta')-1) \right) \nonumber \\
&+ \left[\eta ((r+s)(1-\eta')-1)+(r+s)(1-2\eta')(1-\eta')\right]p_{1,0}\nonumber\\
&+\sum_{\substack{m,n=0 \\ m+n \geq 2}}^M \left[(1-\eta')\left(2(r+s)(1-\eta')-1 \right)(m+n)+(r+s)(1-\eta')(2(1-\eta')-1)+\eta' \right]p_{m,n}.
 \end{align}
The coefficients of $p_{m,n}$ are positive which follows from the assumption $\eta' = 1/(1+\xi') \leq 1/2$. The coefficient of $p_{1,0}$ is non-negative. Also, the remaining term of \eqref{new2_eqn5} is non-negative if
\begin{align}\label{new2_eqn6}
 2E \leq \dfrac{(r+s)(1-\eta')-1}{(1-\eta')(2(r+s)(1-\eta')-1)}.
\end{align}
The right-hand side of equation \eqref{new2_eqn6} is an increasing function of $r+s \geq 2$, which follows because its derivative with respect to $r+s$ is given by $[1-2(r+s)(1-\eta')]^{-2}$ and thus is non-negative for all $r+s$. %See Figure~\ref{plot_constant2_term}.
% \begin{figure}
% \begin{subfigure}{.5\textwidth}
%   \centering
%   \includegraphics[width=.8\linewidth]{figures/constant2_0.1.pdf}
%   \caption{$\eta=0.1$}
%   \label{constant2_0.1}
% \end{subfigure}%
% \begin{subfigure}{.5\textwidth}
%   \centering
%   \includegraphics[width=.8\linewidth]{figures/constant2_0.2.pdf}
%   \caption{$\eta=0.2$}
%   \label{constant2_0.2}
% \end{subfigure}
% \begin{subfigure}{.5\textwidth}
%   \centering
%   \includegraphics[width=.8\linewidth]{figures/constant2_0.3.pdf}
%   \caption{$\eta=0.3$}
%   \label{constant2_0.3}
% \end{subfigure}%
% \begin{subfigure}{.5\textwidth}
%   \centering
%   \includegraphics[width=.8\linewidth]{figures/constant2_0.4.pdf}
%   \caption{$\eta=0.4$}
%   \label{constant2_0.4}
% \end{subfigure}
% \caption{The right-hand side of \eqref{new2_eqn6} for $2 \leq r+s \leq 10$}
% \label{plot_constant2_term}
% \end{figure}
Therefore, a sufficient condition on $E$ for $\partial g_{M,\eta,\eta'}(p)/\partial p_{r,s} > 0$ is obtained from \eqref{new2_eqn6} by substituting $r+s=2$ in the right-hand side.
A sufficient condition is
\begin{align}
 2E &\leq \dfrac{2(1-\eta')-1}{(1-\eta')(4(1-\eta')-1)} \\
 &= \dfrac{1-2\eta'}{(1-\eta')(3-4\eta')} \\
 &= \dfrac{\xi'^2-1}{\xi'(3\xi'-1)},
\end{align}
which holds by hypothesis. So, $g_{M,\eta, \eta'}$ is a strictly increasing function of $p_{r,s}$ for all $r,s$ with $r+s \geq 2$. Therefore, a necessary condition for any of its minimizer $p \in \mathbb{R}^{(M+1)^2}$ is $p_{r,s}=0$ for all $r,s$ with $r+s \geq 2$. 
Also, we know by \eqref{re:bidirectionalfunction} that $ f_{M,\xi, \xi'}=g_{M,\eta, \eta'}$ at such points. Therefore, it follows by \eqref{new2_eqn3} that the minimum value of $f_{M,\eta, \eta'}$ is obtained at a point $p$ for which $p_{r,s}=0$ for all $r,s$ with $r+s \geq 2$. Let $q$ be such a point in the feasible region. From \eqref{new_eqn45} we have $q_{0,0}=1-2E,$
$q_{0,1}=2E-q_{1,0}$. Note that $q_{1,0} \in [0,2E]$. This gives
\begin{align}\label{new2_eqn7}
f_{M,\xi, \xi'}(q) &=g_{M,\eta,\eta'}(q) \\
&= (1-2E)^2+2(1-2E)(\eta' (2E-q_{1,0})+\eta q_{1,0})+2 \eta \eta' (2E-q_{1,0})q_{1,0} \nonumber \\
&\hspace{0.5cm} + (\eta'^2+(1-\eta')^2)(2E-q_{1,0})^2+ (\eta^2+(1-\eta)^2)q_{1,0}^2 \\
&= (1-2E)(1-2E+4\eta'E)+4E^2(\eta'^2 + (1-\eta')^2) \nonumber \\
&\hspace{0.5cm}+ \left[2(1-2E)(\eta-\eta')+4\eta\eta'E-4E(\eta'^2 + (1-\eta')^2)  \right]q_{1,0}+ 2\left[(\eta-\eta')^2+(1-\eta)(1-\eta') \right]q_{1,0}^2\label{new12_eqn7}.
 \end{align}
It is a quadratic polynomial in $q_{1,0},$ and the coefficient of $q_{1,0}^2$ is positive. The global minimum of a quadratic polynomial $c+bx+ax^2$ with $a>0$ is attained by $x=-b/(2a)$. Also, the polynomial is a decreasing function of $x < -b/(2a)$ and an increasing function of $x >-b/(2a)$. So, the minimum of $c+bx+ax^2$ over $[0,2E]$ occurs at
\begin{equation}\label{new2_eqn9}
  x=  \begin{cases}
   0 & \text{ if } -b/2a < 0, \\
   2E & \text{ if } -b/2a  > 2E, \\
   -\dfrac{b}{2a} & \text{ otherwise}.
\end{cases}
\end{equation}
By comparing the quadratic polynomial \eqref{new12_eqn7} with $c+bx+ax^2,$ and from \eqref{new2_eqn9}, the minimizer of $f_{M,\xi, \xi'}$ is given by \eqref{new2_eqn8}.
\end{proof}

%We now find the solution to the truncated problem \eqref{new_eqn30} and use Proposition~\ref{new_prop2} to find
%the solution to the general problem \eqref{new_eqn1} for $\xi' \geq 1$ and arbitrary $\xi$.

\begin{lemma}\label{new2_lem2}
The $M$-truncated energy-constrained channel fidelity \eqref{new_eqn30} has the unique solution given by  \eqref{new5_eqn2},
whenever $\xi' \geq 1$ and $2E \leq \min\{(\xi'^2-1)/(\xi'(3\xi'-1)), (1+\xi)/(2\xi)\}$.
Also, the energy-constrained channel fidelity is
\begin{align}\label{new2_eqn26}
    F_{E,M}(\mathcal{S}_{AB\to B'A'}, \mathcal{T}_{A\to B'}^\xi\otimes \mathcal{T}_{B\to A'}^{\xi'}) =
    \begin{cases}
     1-4\left(\dfrac{\xi'}{1+\xi'}\right)E+8\left(\dfrac{\xi'}{1+\xi'}\right)^2E^2 & \text{ if } 2E < \dfrac{(\xi - \xi') (1+\xi')}{\xi'^2 (1+\xi)}, \\
   1-4\left(\dfrac{\xi}{1+\xi}\right)E+8\left(\dfrac{\xi}{1+\xi}\right)^2E^2 & \text{ if } 2E < \dfrac{(\xi'-\xi)(1+\xi)(1+\xi')}{\xi^2 (1+\xi')^2+ (\xi-\xi')^2}, \\
   \dfrac{4ac-b^2}{4a} & \text{ otherwise},
\end{cases}
\end{align}
where
\begin{align}
    a&=\dfrac{2(\xi'-\xi)^2}{(1+\xi)^2(1+\xi')^2}+\dfrac{2\xi\xi'}{(1+\xi)(1+\xi')},\\
    b&=\dfrac{2(\xi'-\xi)}{(1+\xi')(1+\xi')}+\dfrac{4 \xi' \left(\xi (1-\xi')-2\xi' \right)E}{(1+\xi)(1+\xi')^2},\\
    c&=1-\dfrac{4\xi'E}{(1+\xi')}+\dfrac{8(\xi'E)^2}{(1+\xi')^2}.
\end{align}
 In particular, the the right hand side of \eqref{new2_eqn26} is independent of the truncation parameter $M$.
\end{lemma}
\begin{proof}
Let $\phi_{RAB}$ be any pure state given by $\vert \phi \rangle_{RAB}=  \sum_{m,n=0}^{M} \sqrt{q_{m,n}} \vert m,n\rangle_R \vert m,n\rangle_{AB},$ where $q=(q_{m,n})_{m,n=0}^M$ is a probability vector such that $\sum_{m,n=0}^{M} (m+n) q_{m,n} =2E_0 \leq 2E,$ and let $\psi_{RAB}$ be the pure state given in \eqref{new5_eqn2}.
We know from \eqref{new2_eqn11} that 
\begin{equation}
    F(\mathcal{S}_{AB \to A'B'}(\phi_{RAB}), \mathcal{T}^{\xi}_{A \to B'} \otimes \mathcal{T}^{\xi'}_{B \to A'} (\phi_{RAB})) = f_{M,\xi,\xi'}(q).
\end{equation}
By Proposition~\ref{new2_prop1}, we thus get
\begin{align}\label{new2_eqn25}
      F(\mathcal{S}_{AB \to A'B'}(\phi_{RAB}), \mathcal{T}^{\xi}_{A \to B'} \otimes \mathcal{T}^{\xi'}_{B \to A'} (\phi_{RAB})) \geq    \eta \eta'\left(h_2(E_0)+h_1(E_0)p_{E_0}+h_0(E_0)p_{E_0}^2 \right),
\end{align}
where $\eta=1/(1+\xi),$ $\eta'=1/(1+\xi'),$  $p_{E_0}$ is given by \eqref{new2_eqn10}, and
 $h_0,h_1,h_2$ are polynomial functions defined by
 \begin{align}
     h_0(x)&=2((\eta-\eta')^2+(1-\eta)(1-\eta')),\\
     h_1(x)&=2(\eta-\eta')+4(1-\eta')(2\eta'-1-\eta)x,\\
     h_2(x)&=1-4(1-\eta')x+8(1-\eta')^2x^2.
 \end{align}
The polynomials $h_0(x),h_1(x), h_2(x)$ are decreasing functions of $x$ for $x < 1/(4(1-\eta'))$.
From the given hypothesis, and the relation $\xi'=(1-\eta')/\eta',$ we have
\begin{align}
    E &\leq \dfrac{\xi'^2-1}{2\xi'(3\xi'-1)}\\
    &=\dfrac{1-2\eta'}{2(1-\eta')(3-4\eta')} \\
    & = \dfrac{1}{2(1-\eta')} \left(1- \dfrac{2(1-\eta')}{3-4\eta'} \right)\\
    &< \dfrac{1}{2(1-\eta')} \left(1- \dfrac{2(1-\eta')}{4-4\eta'} \right)\\
    &= \dfrac{1}{4(1-\eta')}.
\end{align}
This implies $h_i(E_0) \geq h_i(E)$ for $i \in \{0,1,2\}$.

Our proof is divided into the three cases, based on the conditions in \eqref{new2_eqn26}. We use
Proposition  \ref{new2_prop1} in each case. 

\emph{Case 1.}
Suppose $2E < (\xi - \xi') (1+\xi')/(\xi'^2 (1+\xi))$.

The minimum value of the right-hand side term of \eqref{new2_eqn25} occurs at $p_{1,0}=0$. We thus get
\begin{align}\label{new2_eqn13}
    F(\mathcal{S}_{AB \to A'B'}(\phi_{RAB}), \mathcal{T}^{\xi}_{A \to B'} \otimes \mathcal{T}^{\xi'}_{B \to A'} (\phi_{RAB})) & \geq \eta \eta' h_2(E_0) \\
    & \geq \eta \eta' h_2(E)\\
    & = F(\mathcal{S}_{AB \to A'B'}(\psi_{RAB}), \mathcal{T}^{\xi}_{A \to B'} \otimes \mathcal{T}^{\xi'}_{B \to A'} (\psi_{RAB})).
\end{align} 

\emph{Case 2.}
Suppose $2E < (\xi'-\xi)(1+\xi)(1+\xi')/(\xi^2 (1+\xi')^2+ (\xi-\xi')^2)$.

The minimum value of the right-hand side term of \eqref{new2_eqn25} occurs at $p_{1,0}=2E$. 
By substituting $p_{1,0}=2E_0$ in \eqref{new2_eqn25} and then by simplifying, we get
\begin{align}\label{new2_eqn14}
     F(\mathcal{S}_{AB \to A'B'}(\phi_{RAB}), \mathcal{T}^{\xi}_{A \to B'} \otimes \mathcal{T}^{\xi'}_{B \to A'} (\phi_{RAB})) \geq \eta \eta' (1-4E_0(1-\eta)+8(1-{\eta})^2 E_0^2).
\end{align}
The polynomial $1-4(1-\eta)x+8(1-{\eta})^2 x^2$ is decreasing for $x \leq 1/(4(1-\eta))$. Since we have $E\leq (1+\xi)/(4\xi)=1/(4(1-\eta)),$ by \eqref{new2_eqn14} we get
\begin{align}
     F(\mathcal{S}_{AB \to A'B'}(\phi_{RAB}), \mathcal{T}^{\xi}_{A \to B'} \otimes \mathcal{T}^{\xi'}_{B \to A'} (\phi_{RAB})) &\geq \eta \eta' ( 1-4(1-\eta)E+8(1-{\eta})^2 E^2) \\
   & =  F(\mathcal{S}_{AB \to A'B'}(\psi_{RAB}), \mathcal{T}^{\xi}_{A \to B'} \otimes \mathcal{T}^{\xi'}_{B \to A'} (\psi_{RAB})).
\end{align}

\emph{Case 3.}
In this case we have
\begin{equation}
2E \geq \max \left\{\frac{(\xi - \xi') (1+\xi')}{\xi'^2 (1+\xi)},  \frac{(\xi'-\xi)(1+\xi)(1+\xi')}{\xi^2 (1+\xi')^2+ (\xi-\xi')^2}\right\}.    
\end{equation}

Now, for all $x \geq 0,$ we have
\begin{align}
    h_2(E_0)+h_1(E_0)x+h_0(E_0)x^2 &\geq  h_2(E)+h_1(E)x+h_0(E)x^2 \label{new13_eqn8} \\ 
    &\geq h_2(E)+h_1(E)p_{E}+h_0(E)p_E^2. \label{new13_eqn9}
\end{align}
In \eqref{new13_eqn8}, we used the fact that $h_i(E_0) \geq h_i(E)$ for all $i \in \{0,1,2\};$  \eqref{new13_eqn9} follows because the global minimum of the polynomial $h_2(E)+h_1(E)x+h_0(E)x^2$ is attained at $x=p_E$. We thus have
\begin{align}\label{new12_eqn8}
    h_2(E_0)+h_1(E_0)p_{E_0}+h_0(E_0)p_{E_0}^2 \geq h_2(E)+h_1(E)p_{E}+h_0(E)p_E^2.
\end{align}
From \eqref{new2_eqn25} and \eqref{new12_eqn8}, we thus get
\begin{align}
     F(\mathcal{S}_{AB \to A'B'}(\phi_{RAB}), \mathcal{T}^{\xi}_{A \to B'} \otimes \mathcal{T}^{\xi'}_{B \to A'} (\phi_{RAB})) &\geq F(\mathcal{S}_{AB \to A'B'}(\psi_{RAB}), \mathcal{T}^{\xi}_{A \to B'} \otimes \mathcal{T}^{\xi'}_{B \to A'} (\psi_{RAB})).
\end{align}

In all the cases, we proved that
\begin{align}
    F(\mathcal{S}_{AB \to A'B'}(\phi_{RAB}), \mathcal{T}^{\xi}_{A \to B'} \otimes \mathcal{T}^{\xi'}_{B \to A'} (\phi_{RAB})) \geq F(\mathcal{S}_{AB \to A'B'}(\psi_{RAB}), \mathcal{T}^{\xi}_{A \to B'} \otimes \mathcal{T}^{\xi'}_{B \to A'} (\psi_{RAB})).
\end{align}
This means that the state in \eqref{new5_eqn2} is optimal. Also, the value of the energy-constrained channel fidelity \eqref{new2_eqn26} can be obtained by direct substitution of the minimizer $p_E$ from \eqref{new2_eqn10}.

\end{proof}

\begin{theorem}\label{main_thm3}
An optimal state for the energy-constrained channel fidelity \eqref{re:bidirectional} is given by \eqref{new5_eqn2},
whenever  $\xi' \geq 1$ and $2E \leq \min\{(\xi'^2-1)/(\xi'(3\xi'-1)), (1+\xi)/(2\xi)\}$. Moreover, the optimal state is unique in the sense that there is no other optimal finite entangled superposition of twin-Fock states. The value of the energy-constrained channel fidelity $F_E(\mathcal{S}_{AB \to A'B'}, \mathcal{T}_{A\to B'}^\xi\otimes \mathcal{T}_{B\to A'}^{\xi'})$ is the same as the value of any of its $M$-truncated counterpart given in \eqref{new2_eqn26}.
\end{theorem}
\begin{proof}
The proof follows similar arguments as given in the proof of Theorem~\ref{main_thm2}, and using Lemma~\ref{new2_lem2}.
\end{proof}

%\nocite{*}

\end{document}